\documentclass[12pt,one column]{IEEEtran}

\usepackage{setspace}
\pdfoutput=1
\usepackage{amsmath,mathtools}
\usepackage{amssymb}
\usepackage{amsthm}
\usepackage{bbm}
\usepackage{dsfont}
\usepackage{cite}
\usepackage{color}
\usepackage{epsfig}
\usepackage{epsf}
\usepackage{rotating}
\usepackage{mathrsfs}
\usepackage{epsfig}
\usepackage{graphics}
\usepackage{enumerate}
\usepackage{accents}
\DeclareMathAccent{\wtilde}{\mathord}{largesymbols}{"65}
\usepackage[T3,T1]{fontenc}
\DeclareSymbolFont{tipa}{T3}{cmr}{m}{n}
\DeclareMathAccent{\invbreve}{\mathalpha}{tipa}{16}
\usepackage[multiple]{footmisc}

\usepackage{anyfontsize}
\usepackage{t1enc}
\usepackage{amssymb}
\usepackage{cite}
\usepackage{color}
\usepackage{epsfig}
\usepackage{epsf}
\usepackage{rotating}
\usepackage{mathrsfs}
\usepackage{epsfig}
\usepackage{graphics}
\usepackage{subfigure}

\newcommand{\vertiii}[1]{{\left\vert\kern-0.25ex\left\vert\kern-0.25ex\left\vert #1 
    \right\vert\kern-0.25ex\right\vert\kern-0.25ex\right\vert}}

\theoremstyle{plain}
\newtheorem{thm}{Theorem}
\newtheorem{lem}{Lemma}
\newtheorem{coro}{Corollary}

\newtheorem{proposition}{Proposition}

\begin{document}
\vskip 1cm

\thispagestyle{empty} \vskip 1cm

%\newpage
%%%%%%%%%%%%%%%%%%%%%%%%%%%%%%%%%%%%%%%%%%%%%%%%%%%%%%%%%%%%%%%%%%%%%%%%%%%%%%%%

\title{{On A Class of Time-Varying  \\Gaussian ISI  Channels }}
\author{ Kamyar Moshksar\\
\small  Vancouver, Canada} %email: kmoshksar@uwaterloo.ca} 
\maketitle

\begin{abstract} 
 This paper studies a class of stochastic and time-varying Gaussian intersymbol interference~(ISI) channels. The probability law for the~$i^{th}$ channel tap during time slot~$t$ is supported over an interval of centre $c_i$ and radius~$ r_{i}$. The transmitter and the receiver only know the centres $c_i$ and the radii $r_i$.  The joint distribution for the array of channel taps and their realizations are unknown to both the transmitter and the receiver. A lower bound (achievability result) is presented on the channel capacity which results in an upper bound on the capacity loss compared to when all radii are zeros. The lower bound on the channel capacity saturates at a positive value as the maximum average input power $P$ increases beyond what is referred to as the saturation power $P_{sat}$. Roughly speaking, $P_{sat}$ is inversely proportional to the sum of the squares of the radii $r_i$.  A partial converse result is provided in the worst-case scenario where the array of channel taps varies independently along both indices $t$ and $i$ with uniform marginals. It is shown that for every sequence of codebooks with vanishing probability of error, if the size of each symbol in every codeword is bounded away from zero by an amount proportional to $\sqrt{P}$, then the rate of that sequence of codebooks does not scale with~$P$. Tools in matrix analysis such as matrix norms and Weyl's inequality on perturbation of eigenvalues of symmetric matrices are used in order to analyze the probability of error. 
\end{abstract}

\section{Introduction}
\subsection{Summary of prior art}
Transmission beyond the Nyquist rate over a band-limited communication channel results in a phenomenon known as intersymbol interference (ISI). The time-invariant Gaussian ISI channel is modelled by a linear filter with known finite impulse-response and additive white Gaussian noise~\cite{Proakis}. The capacity for this channel was first studied in \cite{galager} where it was shown that Gaussian signalling is capacity-achieving. Addressing more practical structures for signal transmission, the Gaussian ISI channel with independent and identically distributed~(i.i.d.) signalling over a fixed finite alphabet was examined in~\cite{Shamai-Laroia} and more recently in~\cite{Shamai-Carmon1,Shamai-Carmon2}. 

%The problem of universal decoding for the Gaussian ISI channel was raised for the first time in~\cite{Merhav1}. Roughly speaking, a universal decoder is one whose structure does not depend on the channel parameters and yet it achieves the same error exponent as the optimum maximum likelihood decoder tuned to the actual channel parameters.  Influenced by the pioneering work of Ziv in~\cite{Ziv} on universal decoding over channels with memory, the authors in~\cite{Feder-Lapidoth} coined various notions of universal decoding such as random vs. deterministic and weak vs. strong universal decoding. In particular, \cite{Feder-Lapidoth} provided an answer to the problem raised in~\cite{Merhav1}. The authors showed that if the receiver knows the length and an upper bound on the energy of the channel impulse response, the Gaussian ISI channel admits a universal decoder in the strong sense. Later in~\cite{Feder-Merhav},  another universal decoder was proposed from the viewpoint of composite hypothesis testing. More recently, a universal decoder was presented in~\cite{Huleihel-Merhav} which is easier to evaluate compared to the previously reported decoders in \cite{Feder-Lapidoth, Feder-Merhav}.

The time-varying Gaussian ISI channel is less studied in the literature. Reference \cite{suhas} presents information-theoretic considerations for multi-carrier transmission in time-varying Gaussian ISI channels. Reference~\cite{goldsmith} derives the capacity of block-memoryless channels with block-memoryless side information. Subsequently, the capacity of time-varying Gaussian ISI channels is characterized. Both~\cite{suhas} and \cite{goldsmith} assume that perfect channel state information about the time-varying channel is available at the receiver side.  

To the author's best knowledge, there are no capacity results on time-varying Gaussian ISI channels in the absence of channel state information at both ends of communication. In the context of time-varying fading channels, perhaps the most relevant work is the landmark paper~\cite{abou-faycal} by Abou-Faycal and others. The authors study a Rayleigh fading channel subject to a constraint on the average input power where the channel coefficient (the channel filter has only one tap) varies independently from symbol to symbol and both the transmitter and the receiver do not have channel state information. It is shown that the capacity-achieving input distribution is discrete with a mass point located at zero.      

\subsection{Summary of contributions}
%We study a class of stochastic and time-varying Gaussian ISI channels where the transmitter and the receiver only know the statistics of the array of channel taps. 
Consider a Gaussian ISI channel initially modelled by a sequence of filter taps $c_0, c_1, \cdots, c_k$ as its impulse response. Let $\boldsymbol{h}_{t,i}$ be the actual $i^{th}$ channel tap during time slot $t$.  Then the initial model assumes $\boldsymbol{h}_{t,i}=c_i$ for all $t$ and $i$. We ask the following question: 

\textit{If $\boldsymbol{h}_{t,i}$ for $t=1,2,\cdots$ are not exactly the constant $c_i$, but a random process (with possibly unknown dynamics) that takes values near $c_i$, then how would that impact the channel capacity? } %The joint probability distribution of the random array $\boldsymbol{h}_{t,k}$ is otherwise unknown to both ends. %We reiterate that the realizations of the random variables $\boldsymbol{h}_{t,k}$ are unknown to both ends \textit{a priori}.  

In an attempt to answer this question, it is assumed that the probability law for $\boldsymbol{h}_{t,i}$ is supported over the interval $[c_i-r_i, c_i+r_i]$ for all $t$ and $i$ where the centres $c_i$ and the radii $r_i\ge0$ are known constants to the transmitter and the receiver. The joint distribution for the array $\boldsymbol{h}_{t,i}$ (including its marginals) as well as the realizations of $\boldsymbol{h}_{t,i}$ remain unknown to both the transmitter and the receiver.  The capacity $C^*$ for this channel is studied subject to a maximum average transmission power~$P$. We investigate the performance of the ensemble of Gaussian codebooks where the codewords are independent zero mean Gaussian vectors with common covariance matrix $\Sigma_n$. Here, $n$ denotes the length of the codewords. 
It is shown that if the sequence of positive-definite matrices $\Sigma_n$ satisfies the fairly mild condition $\lim_{n\to\infty}\frac{1}{n}\lambda_{\max}(\Sigma_n)=0$ where $\lambda_{\max}(\Sigma_n)$ is the largest eigenvalue of $\Sigma_n$,  then the probability of decoding error is guaranteed to vanish as~$n$ grows to infinity for sufficiently small values of the transmission rate. As a result, we obtain a lower bound $C^*_{LB}$ on the capacity~$C^*$ which is presented in Theorem~1 in Section~II. This lower bound is further maximized in Proposition~1 in Section~III over the admissible set of matrices~$\Sigma_n$ where water-filling is performed. A feature of $C^*_{LB}$ is that it saturates at a positive value as $P$ increases beyond what we refer to as the saturation power denoted by $P_{sat}$. Its value is given by $P_{sat}=\frac{2}{(k+1)\sum_{i=0}^k r_i^2}-2J$ where $J$ depends entirely on the coefficients $c_i$. An upper bound is established in Corollary~2 in Section~III on the amount of loss in channel capacity before saturation occurs compared to when all radii~$r_i$ are zeros. %To establish such saturation, a new upper bound is derived on the size of a Gaussian typical set in Appendix~ which improves on the existing and well-known upper bound given in~\cite{Cover-Thomas}. 
We also provide a partial converse result in Theorem~2 in Section~II. It addresses the worst-case scenario where the array $\boldsymbol{h}_{t,i}$ varies independently along both indices $t$ and $i$ and its marginals are uniform. % which is in agreement with our achievability result. 
It states that for a sequence of codebooks with vanishing probability of error, if there exists a constant $a>0$ such that the size of every symbol in each codeword is at least~$a\sqrt{P}$, then the rate of that sequence of codebooks does not scale with~$P$. Tools in matrix analysis such as matrix norms and Weyl's inequality on perturbation of eigenvalues of symmetric matrices are used in order to analyze the probability of error. %Special attention is given to the scenario where the array $\boldsymbol{h}_{t,k}$ is independent. %We also present an upper bound on $C_{\scriptscriptstyle{WGI}}$. The two bounds are compared in a few examples.  Moreover, the transmitter and the receiver do not know the actual channel realizations \textit{a priori}. 

The rest of the paper is organized as follows. We end the current section by a list of adopted notations. System model, the problem statement and a summary of main results appear in Section~II. Section~III further explores the lower bound presented in Theorem~1 in Section II. Section~IV describes the structure of the proposed decoder. Section~V is devoted to error analysis where we establish the main result (Theorem~1). The majority of details of the proofs are deferred to the appendices at the end of the paper. Finally, Section~VI concludes the paper.

\subsection{Notations}
For a real number $x$, $x^+=\max\{0,x\}$. The Euclidean space of sequences of length $n$ whose entries are real numbers is denoted by $\mathbb{R}^n$. Vectors and sequences are identified by an underline such as $\underline{x}$.  Random quantities appear in bold such as $\boldsymbol{x}$ and $\underline{\boldsymbol{x}}$ with realizations $x$ and $\underline{x}$, respectively. An $m\times n$ matrix whose all entries are zeros is denoted by $0_{m,n}$. We use $\underline{0}_{n}$ to denote a column vector of length $n$ whose all entries are zeros. The $n\times n$ identity matrix is denoted by~$I_n$. The transpose, inverse, trace and determinant of a square matrix $A$ are denoted by $A^T, A^{-1}, \mathrm{tr}(A)$ and $\det (A)$, respectively. The entry at the $i^{th}$ row and the $j^{th}$ column of a matrix $A$ is denoted by $A_{i,j}$ or $[A]_{i,j}$. The $i^{th}$ entry of a vector $\underline{x}$ is denoted by $x_i$ or $[\underline{x}]_i$. The 2-norm~(Frobenius norm) and the spectral norm~(operator~norm)  of a matrix $A$ are defined by 
\begin{eqnarray}
\|A\|_2:=(\mathrm{tr}(A^TA))^{1/2}
\end{eqnarray}
and  
\begin{eqnarray}
\label{opn_111}
\|A\|:=\max_{\|\underline{x}\|_2\leq 1}\|A\underline{x}\|_2=\big(\lambda_{\max}(A^TA)\big)^{\frac{1}{2}},
\end{eqnarray}
respectively, where $\lambda_{\max}(M)$ denotes the largest eigenvalue of a symmetric matrix $M$. The smallest eigenvalue of a symmetric matrix $M$ is denoted by $\lambda_{\min}(M)$. The probability law of a random variable $\boldsymbol{x}$ (the induced probability measure on the range of $\boldsymbol{x}$) is denoted by $\mathcal{L}_{\boldsymbol{x}}(\cdot)$ and its mean is denoted by~$\mathbb{E}[\boldsymbol{x}]$. The probability density function~(PDF) of a continuous random variable $\underline{\boldsymbol{x}}$ is denoted by $p_{\boldsymbol{x}}(\cdot)$. The PDF of a Gaussian random vector of length $n$ with zero mean and covariance matrix $\Sigma$~is~denoted~by 
 \begin{eqnarray}
p_{\mathsf{G}}(\underline{x};\Sigma): =\frac{1}{(2\pi)^\frac{n}{2}\sqrt{\det(\Sigma)}}\exp\big(-\frac{1}{2}\underline{x}^T\Sigma^{-1}\underline{x}\big).
\end{eqnarray}
 We write $\underline{\boldsymbol{x}}\sim \mathrm{N}(\underline{0}_{n}, \Sigma)$ if $p_{\boldsymbol{\underline{x}}}(\underline{x})=p_{\mathsf{G}}(\underline{x};\Sigma)$. The differential entropy of a random vector $\underline{\boldsymbol{x}}\sim \mathrm{N}(\underline{0}_{n}, \Sigma)$ is denoted by $h_{\mathsf{G}}(\Sigma)$ given by 
 \begin{equation}
\label{ }
h_{\mathsf{G}}(\Sigma):=\frac{1}{2}\log\big((2\pi e)^n\det(\Sigma)\big),
\end{equation}
where $\log(\cdot)$ is the logarithm function with base $2$. We also recall the definition for a typical set~\cite{Cover-Thomas}. For $\eta>0$, positive integer~$n$ and an $n\times n$ positive-definite matrix $\Sigma$, the Gaussian typical set $\mathcal{T}_{\eta}^{(n)}(\Sigma)$ is defined by 
\begin{eqnarray}
\label{fgb_111}
\mathcal{T}_{\eta}^{(n)}(\Sigma)&:=&\Big\{\underline{a}\in\mathbb{R}^n: \Big|\frac{1}{n}\log p_{\mathsf{G}}(\underline{a};\Sigma)+\frac{1}{n}h_{\mathsf{G}}(\Sigma)\Big|<\frac{\eta}{2\ln 2}\Big\}\notag\\&=&\Big\{\underline{a}\in\mathbb{R}^n: \Big|\frac{1}{n}\underline{a}^T \Sigma^{-1}\underline{a}-1\Big|<\eta\Big\}.
\end{eqnarray}
%The underlying probability measure is denoted by $\Pr(\cdot)$. Finally, $\mathbb{E}[\cdot]$ denotes the expectation operator against $\Pr(\cdot)$.  
\section{System model and summary of results}
 The stochastic and time-varying Gaussian ISI channel with memory-depth $k\geq 1$ is defined by
\begin{eqnarray}
\label{model_1}
\boldsymbol{y}_t=\sum_{i=0}^{k}\boldsymbol{h}_{t,i}\boldsymbol{x}_{t-i}+\boldsymbol{z}_t,\,\,\,t=1,2,\cdots,
\end{eqnarray}
where $\boldsymbol{h}_{t,0},\boldsymbol{h}_{t,1},\cdots,\boldsymbol{h}_{t,k}$ are the channel coefficients (filter taps) during time slot $t=1,2,\cdots$  and $\boldsymbol{x}_t$, $\boldsymbol{y}_t$~and~$\boldsymbol{z}_t$ are the channel input, the channel output and the additive noise at time slot $t$, respectively.  The process $\boldsymbol{z}_1, \boldsymbol{z}_2,\cdots$ is white Gaussian with zero mean and unit variance. The input process $\boldsymbol{x}_1, \boldsymbol{x}_2,\cdots$ is subject to the average power constraint 
\begin{eqnarray}
\label{ghor_1}
\sum_{t=1}^{n}\mathbb{E}[|\boldsymbol{x}_{t}|^2]\leq nP,
\end{eqnarray}
where $n$ is the length of the communication period of interest. 

The following assumptions are made on the array of channel coefficients:  
\begin{enumerate}[(i)]
\item %The array of channel taps $\boldsymbol{h}_{t,i}$ is independent along both indices $t$ and $i$ with uniform marginals. More precisely, 
There is a sequence of numbers $\underline{c}=(c_0,c_1,\cdots, c_k)$ and a sequence of nonnegative numbers~$\underline{r}=( r_0,  r_1,\cdots, r_k)$  such that  the probability law for $\boldsymbol{h}_{t,i}$ is supported over the interval $[c_i-r_i,c_i+r_i]$  for every $t,i$. We denote the sum of the radii $r_i$ by $r_s$, i.e., 
\begin{eqnarray}
\label{ehj_1}
r_s:=\sum_{i=0}^k r_i.
\end{eqnarray}
%It is only assumed that the marginals are uniform in the aforementioned intervals.  The joint distribution of  $\boldsymbol{h}_{t,k}$ is otherwise arbitrary.  
\item The transmitter and the receiver only know the centres $c_i$ and the radii $r_i$. The joint distribution of the random variables $\boldsymbol{h}_{t,i}$ as well as their realizations are unknown to both ends of communication. %The realizations of $\boldsymbol{h}_{t,k}$ are unknown to both ends \textit{a priori}.   
\item The processes $\boldsymbol{h}_{t,i}$ and $\boldsymbol{z}_t$ are independent. %The random array $\boldsymbol{h}_{t,k}$ is not necessarily stationary or ergodic.  
 \end{enumerate} 
 
 To describe the ISI channel in matrix form, assume a codeword is transmitted during time slots $t=1,\cdots, n$. Define%\footnote{We make the dependence of $m$ on $n$ implicit for notational simplicity. } 
\begin{eqnarray}
m:=n+k
\end{eqnarray}
and
\begin{eqnarray}
\label{y_1}
\underline{\boldsymbol{x}}:=\big[\boldsymbol{x}_1\,\,\,\boldsymbol{x}_2\,\,\,\cdots\,\,\,\boldsymbol{x}_n\big] ^T,\,\,\,\,\,\,\,\underline{\boldsymbol{y}}:=\big[\boldsymbol{y}_1\,\,\,\boldsymbol{y}_2\,\,\,\cdots\,\,\,\boldsymbol{y}_m\big]^T,\,\,\,\,\,\,\,\,\,\underline{\boldsymbol{z}}:=\big[\boldsymbol{z}_1\,\,\,\boldsymbol{z}_2\,\,\,\cdots\,\,\,\boldsymbol{z}_m\big]^T.
\end{eqnarray}
 Then 
 \begin{eqnarray}
 \label{jat_123}
\underline{\boldsymbol{y}}=\boldsymbol{H}\underline{\boldsymbol{x}}+\underline{\boldsymbol{z}},
\end{eqnarray}
 where the $m\times n$ random channel matrix $\boldsymbol{H}$ is given by 
 \begin{eqnarray}
 \label{gg_hh}
\boldsymbol{H}_{i,j}:=\left\{\begin{array}{cc}
    \boldsymbol{h}_{i,i-j}  & 0\le i-j\leq k   \\
     0 &   \mathrm{otherwise}
\end{array}\right..
\end{eqnarray}
We also define the $m\times n$ deterministic matrix $H_c$ by 
  \begin{eqnarray}
  \label{vrty}
[H_c]_{i,j}:=\left\{\begin{array}{cc}
    c_{i-j}  & 0\le i-j\leq k   \\
     0 &   \mathrm{otherwise}
\end{array}\right..
\end{eqnarray}
For example, when $n=4$ and $k=2$, we have $m=4+2=6$ and the matrices $\boldsymbol{H}$ and $H_c$ are given~by 
\begin{eqnarray}
\boldsymbol{H}=\begin{bmatrix}
     \boldsymbol{h}_{1,0} & 0 &0 & 0 \\
     \boldsymbol{h}_{2,1} &\boldsymbol{h}_{2,0}  &0 & 0   \\
     \boldsymbol{h}_{3,2} & \boldsymbol{h}_{3,1} & \boldsymbol{h}_{3,0} & 0 \\
     0 &\boldsymbol{h}_{4,2} & \boldsymbol{h}_{4,1}& \boldsymbol{h}_{4,0}\\
     0 & 0& \boldsymbol{h}_{5,2}&\boldsymbol{h}_{5,1} \\
     0 &0 & 0 &\boldsymbol{h}_{6,2}
\end{bmatrix},\hskip1cm  H_c=\begin{bmatrix}
     c_{0} & 0 &0 & 0 \\
     c_{1} &c_{0}  &0 & 0   \\
     c_{2} & c_{1} & c_{0} & 0 \\
     0 &c_{2} & c_{1}& c_{0}\\
     0 & 0& c_{2}&c_{1} \\
     0 &0 & 0 & c_{2}
\end{bmatrix}.
\end{eqnarray} 
Throughout the paper, we will denote the range of the random matrix $\boldsymbol{H}$ by $\mathcal{H}$, i.e., $\mathcal{H}$ is the set of all $m\times n$ matrices $H$ such that $H_{i,j}=0$ for $i-j<0$ or $i-j>k$ and $|H_{i,j}-c_{i-j}|\leq r_{i-j}$ for $0\leq i-j\leq k$.

The message $\boldsymbol{W}$ is uniformly distributed over the set of indices $\{1,2,\cdots, 2^{nR}\}$ where $R$ is the transmission rate. The encoder  maps  $\boldsymbol{W}$ to a codeword $\underline{\boldsymbol{x}}$ of length $n$. This codeword is then transmitted over the channel in (\ref{jat_123}) during time slots~$t=1,\cdots, n$. The decoder receives the vector $\underline{\boldsymbol{y}}$ and generates the estimate $\hat{\boldsymbol{W}}$ for $\boldsymbol{W}$. The probability of decoding~error~is denoted~by
\begin{eqnarray}
p_{e,n}:=\Pr(\hat{\boldsymbol{W}}\neq \boldsymbol{W}).
\end{eqnarray}
%where we have made its dependence on the underlying codebook $\mathcal{C}$ explicit.
We say a transmission rate~$R$ is achievable if there exists a sequence of encoder-decoder pairs such that  
\begin{eqnarray}
\label{ens_00}
%\lim_{n\to\infty}\mathbb{E}\big[p_e(\boldsymbol{\mathcal{C}})\big]=0,
\lim_{n\to\infty} p_{e,n}=0.
\end{eqnarray}
%where the random codebook $\boldsymbol{\mathcal{C}}$ comes from the ensemble $\mathcal{C}$. 
The supremum of all achievable rates is denoted by $C^*=C^*( \underline{c},\underline{r}, P)$ and referred to as the capacity of the stochastic and time-varying ISI channel with parameters $\underline{c}$ and $\underline{r}$ under a maximum average power~of~$P$. 

Let
 \begin{eqnarray}
f(\omega):=\sum_{l=0}^{k}c_le^{\sqrt{-1}\,l\omega}
\end{eqnarray}
be the discrete Fourier transform of the~sequence~$\underline{c}$. In the special case where the ISI channel is deterministic and time-invariant, i.e., $ r_{i}=0$ for every~$i$, it is well-known\footnote{See \cite{Hirt-Massey} and the references therein.} that the capacity, denoted in this case by $C^*_0=C^*( \underline{c},\underline{0}, P)$, is given~by 
\begin{eqnarray}
\label{known_1}
C^*_0=\frac{1}{4\pi}\int_0^{2\pi}\log[\max(\Theta |f(\omega)|^2,1)]d\omega,
\end{eqnarray}
where the parameter $\Theta$ solves 
\begin{eqnarray}
\label{faghat}
\frac{1}{2\pi}\int_0^{2\pi}(\Theta-|f(\omega)|^{-2})^+d\omega=P.
\end{eqnarray}
% For the time-varying channel where $ r_{k}>0$ for every $k\ge0$, the inability of the receiver in tracking $\boldsymbol{h}_{t,k}$ for $t=1,2,\cdots$ comes at the cost of a rate loss, i.e., 
% \begin{eqnarray}
%C_{\scriptscriptstyle{WGI}}\leq C_{\scriptscriptstyle{WGI},0}.
%\end{eqnarray}
 % In~fact,  a standard application of Fano's inequality shows that 
 %\begin{eqnarray}
 %\label{inverse}
%C_{\scriptscriptstyle{WGI}}\leq \inf_{\underline{h}\in \mathcal{B}(\underline{c}, \underline{r}) }C_{\scriptscriptstyle{WGI}}(\underline{h},\underline{0}),
%\end{eqnarray}
%where  the ball $\mathcal{B}(\underline{c},\underline{r})$ is the set of all sequences $\underline{h}=(h_0,h_1,\cdots)$ such that $|h_k-c_k|\leq  r_{k}$ for all $k\geq0$. The details are provided in Appendix~C for completeness. 
 The main contributions of this paper are a lower bound (achievability result) on $C^*$ for arbitrary~$\underline{c}, \underline{r},P$ and a partial converse result. To present these results, fix a sequence of positive-definite matrices $\Sigma_n$ such that\footnote{As a non-example, the matrices $\Sigma_n$ whose all diagonal entries are $P$ and whose all entries off the main diagonal are $\rho P$ for some $\rho\neq 0$ satisfy (\ref{sx_1}), but they do not satisfy (\ref{sx_2}).} 
\begin{eqnarray}
\label{sx_1}
\mathrm{tr}(\Sigma_n)\leq nP
\end{eqnarray}
and 
\begin{eqnarray}
\label{sx_2}
\lim_{n\to\infty}\frac{1}{n}\lambda_{\max}(\Sigma_n)=0.
\end{eqnarray}
Define 
\begin{eqnarray}
\label{min_max}
\alpha:=\min_{\omega\in[0,2\pi]}|f(\omega)|,\,\,\,\,\,\,\, \beta:=\max_{\omega\in[0,2\pi]}|f(\omega)|,
\end{eqnarray}
%If there exists an (unique) index $0\le i_0\leq k$ such that 
%\begin{eqnarray}
 %\frac{1}{6}Pr^2_{i_0}>1+\frac{1}{6}P\sum_{i\neq i_0}r_i^2,
%\end{eqnarray}
%then define 
%\begin{eqnarray}
%\label{grid_opt}
%\phi_1:=\min_{L\in \mathbb{N}}\,L\Big(1+\frac{1}{6}P\sum_{i\neq i_0}r_i^2+\frac{ \frac{1}{6}Pr^2_{i_0}}{L^2}\Big),
%\end{eqnarray}
%If there is no such index $i_0$, let
%\begin{eqnarray}
%\phi_1:=1+\frac{1}{6}P\|\underline{r}\|_2^2.
%\end{eqnarray}
\begin{eqnarray}
\label{joda_22}
\phi_{1,n}:= \frac{r_s( r_s+2\beta)\lambda_{\max}(\Sigma_n)}{1+\alpha^2\lambda_{\min}(\Sigma_n)}
\end{eqnarray}
and\footnote{The parameters $\phi_{1,n}, \phi_{2,n}$ and $\phi_{3,n}$ appear in the study of the type II error in Section~V. }
\begin{eqnarray}
\label{nmkio}
\phi_{2,n}:=r_s(r_s+2\beta)\frac{\mathrm{tr}(\Sigma_n)}{m},\,\,\,\,\,\phi_{3,n}:=\frac{1}{1+ r_s(r_s+2\beta)\lambda_{\max}(\Sigma_n)}.
\end{eqnarray}
 %where $r_s$ is defined in (\ref{ehj_1}).
%Note that $\phi$ can be made arbitrarily small by letting $ r_s$ be small enough. % 
We are ready to present the achievability result.
\newpage
\begin{thm}
Let $\Sigma_n$ be a sequence of positive-definite matrices that satisfies both (\ref{sx_1}) and~(\ref{sx_2}). Assume $r_s$ is small enough such that $\phi_{1,n}<1$ for all $n$. Then 
\begin{eqnarray}
\label{lb_00}
C^*\geq C^*_{LB}:=\liminf_{n\to\infty}\bigg(\frac{1}{2n}\log\det(I_n+H_c^TH_c\Sigma_n)-\log\Big(1+(k+1)\|\underline{r}\|_2^2\frac{\mathrm{tr}(\Sigma_n)}{m+n}\Big)-\delta_n\bigg),
\end{eqnarray} 
where 
\begin{eqnarray}
\label{kool}
\delta_n:=-\frac{1}{2}\log(1-\phi_{1,n})+\frac{1}{2\ln 2}(1-(1-\phi_{2,n})^+\phi_{3,n}).
\end{eqnarray} 
%\begin{eqnarray}
%\label{rho_00}
%\varrho(\tau):=\left\{\begin{array}{cc}
%    -\frac{1}{2}\log(1-\tau)  &  \tau<1  \\
%    \infty  &   \tau\ge1
%\end{array}\right.
%\end{eqnarray}
% and $\phi_1, \phi_2$, $\phi_3, \phi_4$ and $\phi_5$ are given in (\ref{joda_11}), (\ref{joda_22}) and (\ref{stela}), respectively.
 %Moreover, if $\boldsymbol{h}_{t,k}$ is an independent array, then one can replace $\xi$ and $\xi'$ by $\frac{1}{6}\|\underline{r}\|_2^2P$ and $\frac{2}{\sqrt{3}}\|\underline{r}\|_2\varrho P$, respectively. 
\end{thm} 
\begin{proof}
The proof is detailed in Section~V.
\end{proof}
%The lower bound in (\ref{lb_00}) is achieved by a joint-typicality decoder that is tuned to $H_c=\mathbb{E}[\boldsymbol{H}]$.    
%We remark that the proposed transmitter and the receiver that achieve the lower bound on $C_{\scriptscriptstyle{WGI}}$ given in (\ref{lb_00}) do not use their knowledge of the joint distribution for $\boldsymbol{h}_{t,k}$.
The lower bound $C^*_{LB}$ is achieved by a Gaussian joint-typicality decoder that is tuned to the matrix $H_c$ rather to the unknown channel matrix $\boldsymbol{H}$. In the absence of channel variations, i.e., when $r_0,r_1,\cdots, r_k$ are all equal to zero, one only gets the first term $\frac{1}{2n}\log\det(I_n+H_c^TH_c\Sigma_n)$ in (\ref{lb_00}). The second term $\log(1+(k+1)\|\underline{r}\|_2^2\frac{\mathrm{tr}(\Sigma_n)}{m+n})$ and the third term $\delta_n$ in~(\ref{lb_00}) come to life in the presence of channel variations. They appear in the study of the so-called type~I and type~II error probabilities, respectively. The details are given in Section~V. 

We also present a partial converse result which is stated next. It addresses the worst-case scenario where the array $\boldsymbol{h}_{t,i}$ varies independently along both indices $t$ and $i$ and its marginals are uniform. This is called the worst-case scenario in the sense that the uniform distribution on a given interval has the largest entropy among all continuous distributions with a density supported on that interval.   

\begin{thm}
Assume the array of channel taps $\boldsymbol{h}_{t,i}$ is independent along both indices $t$ and $i$ and $\boldsymbol{h}_{t,i}$ is uniformly distributed on $[c_i-r_i,c_i+r_i]$. Let $\mathcal{C}_n=\{\underline{x}_{n,1},\cdots, \underline{x}_{n,2^{nR}}\}$ be a sequence of codebooks with rate $R$ and vanishingly small probability of error that satisfy the average power constraint in (\ref{ghor_1}). Then 
\begin{eqnarray}
\label{ub_ub}
R\leq \frac{1}{2}\log\Big(1+(k+1)\big(\|\underline{c}\|_2^2+\|\underline{r}\|_2^2/3\big)P\Big)-\kappa,
\end{eqnarray}
where
\begin{eqnarray}
\kappa:=\liminf_{n\to\infty}\frac{1}{n2^{nR}}\sum_{i=1}^{2^{nR}}\sum_{t=1}^{n+k} \frac{1}{2}\log\Big(1+\frac{2}{\pi e}\sum_{j=0}^{k}r^2_{j}x^2_{n,i,t-j}\Big).
\end{eqnarray} 
Here, $x_{n,i,t-j}$ is the $(t-j)^{th}$ symbol of the $i^{th}$ codeword $\underline{x}_{n,i}$ in the $n^{th}$ codebook $\mathcal{C}_n$.
\end{thm}
\begin{proof}
See Appendix~A.
\end{proof}
By Theorem~2, if the size of each symbol in every codeword is larger than a constant $x_{\min}>0$, then $\kappa\geq \frac{1}{2}\log(1+\frac{2}{\pi e}\|\underline{r}\|_2^2 x^2_{\min})$ and every achievable rate $R$ satisfies
\begin{eqnarray}
\label{erft}
R\leq\frac{1}{2}\log\frac{1+(k+1)\big(\|\underline{c}\|_2^2+\|\underline{r}\|_2^2/3\big)P}{1+\frac{2}{\pi e}\|\underline{r}\|_2^2\,x_{\min}^2}.
\end{eqnarray}
 If the codewords are constructed over an alphabet that is included inside $(-\infty,-a\sqrt{P}]\cup [a\sqrt{P},\infty)$ for some constant $a>0$, then $x_{\min}\geq a\sqrt{P}$ and the expression on the right side of (\ref{erft}) does not scale with $P$. It is shown in the next section that the lower bound $C^*_{LB}$ in Theorem~1 also saturates at a positive level as $P$ increases regardless of $k\geq 1$. We emphasize that $C^*_{LB}$ is achievable under any joint distribution (not just the I.I.D. model) for the random array $\boldsymbol{h}_{t,i}\in[c_i-r_i,c_i+r_i]$. 
%Both $C_{\scriptscriptstyle{WGI}}^{(\mathrm{lb})}$ and $C_{\scriptscriptstyle{WGI}}^{(\mathrm{ub})}$  simplify to $C_{\scriptscriptstyle{WGI},0}(P)$ when $r_0=r_1=\cdots=r_k=0$. 

 %The term $\mathbb{E}\big[\log \chi^2_{k+1}\big]$ can be written in terms of the digamma function as 
%\begin{eqnarray}
%\mathbb{E}\big[\log \chi^2_{k+1}\big]=
%\end{eqnarray}
 %In this case, $\kappa=\log r+\frac{1}{2}\log k$. 
\section{Exploring The Lower Bound in Theorem~1 }
The matrix $H_c^TH_c$ is a symmetric Toeplitz matrix whose $(i,j)$-entry is given by 
\begin{eqnarray}
[H_c^TH_c]_{i,j}=\left\{\begin{array}{cc}
     \sum_{l=0}^{k-|i-j|}c_lc_{l+|i-j|} & |i-j|\leq k   \\
   0   &  \mathrm{otherwise} 
\end{array}\right..
\end{eqnarray}
Let $H_c^TH_c=U\Lambda U^T$ be the eigenvalue decomposition of $H_c^TH_c$ where $U$ is an $n\times n$ orthogonal matrix and the diagonal matrix $\Lambda$ carries the eigenvalues 
\begin{eqnarray}
\lambda_{n,1}\leq \cdots\leq\lambda_{n,n}
\end{eqnarray}
of $H_c^TH_c$ on its diagonal. We choose  
\begin{eqnarray}
\Sigma_n=U^TDU,
\end{eqnarray} 
where 
\begin{eqnarray}
D=\mathrm{diag}(d_{n,1},\cdots, d_{n,n}),
\end{eqnarray}
and $d_{n,1},\cdots, d_{n,n}$ satisfy 
\begin{eqnarray}
\label{vgt_11}
\sum_{i=1}^n d_{n,i}\leq nP
\end{eqnarray}
and 
\begin{eqnarray}
\label{vgt_22}
\lim_{n\to\infty}\frac{1}{n}\max_{1\le i\le n} d_{n,i}=0.
\end{eqnarray}
The conditions in (\ref{vgt_11}) and (\ref{vgt_22}) are the ones in (\ref{sx_1}) and (\ref{sx_2}), respectively.  Then $C^*_{LB}$ can be written as  
\begin{eqnarray}
\label{khand}
C^*_{LB}=\liminf_{n\to\infty}\Big(\frac{1}{2n}\sum_{i=1}^n\log(1+\lambda_{n,i}d_{n,i})-\log\Big(1+(k+1)\|\underline{r}\|_2^2\frac{\sum_{i=1}^nd_{n,i}}{m+n}\Big)-\delta_n\Big).
\end{eqnarray}
Since $\lambda_{n,i}$ is the coefficient of $d_{n,i}$ in $\sum_{i=1}^n\log(1+\lambda_{n,i}d_{n,i})$, we let
\begin{eqnarray}
\label{vgt_33}
0<d_{n,1}\leq d_{n,2}\leq \cdots\leq d_{n,n},
\end{eqnarray}
to ensure $C^*_{LB}$ is the largest. As mentioned earlier, we require $\Sigma_n$ be positive definite and hence, all eigenvalues of $\Sigma_n$ must be positive. The term $\delta_n$ in (\ref{kool}) is given by
\begin{eqnarray}
\label{pookl}
\delta_n&=&-\frac{1}{2}\log\Big(1-\frac{r_s(r_s+2\beta)d_{n,n}}{1+\alpha^2d_{n,1}}\Big)\notag\\
&&+\frac{1}{2\ln2}\Big(1-\Big(1-r_s(r_s+2\beta)\frac{\sum_{i=1}^nd_{n,i}}{m}\Big)^+\frac{1}{1+r_s(r_s+2\beta)d_{n,n}}\Big).
\end{eqnarray}
The lower bound in (\ref{khand}) can be maximized over all choices of $d_{n,1},\cdots, d_{n,n}$ that satisfy the conditions in (\ref{vgt_11}), (\ref{vgt_22}) and (\ref{vgt_33}). For simplicity, we will look into the problem of maximizing $\frac{1}{2n}\sum_{i=1}^n\log(1+\lambda_{n,i}d_{n,i})-\log(1+(k+1)\|\underline{r}\|_2^2\frac{\sum_{i=1}^nd_{n,i}}{m+n})$, i.e., we do not involve $\delta_n$ in the optimization process. The first order KKT conditions are investigated in Appendix~B where the next proposition is established.
\begin{proposition}
\label{prop_11}
Let $\Theta_1$ be the solution to 
\begin{eqnarray}
\label{pre_ziba_1}
 \frac{1}{2\pi}\int_0^{2\pi}(\Theta-|f(\omega)|^{-2})^+\mathrm{d}\omega=P
\end{eqnarray}
and $\Theta_2$ be the solution to 
\begin{eqnarray}
\label{pre_ziba_2}
\frac{1}{2\pi}\int_0^{2\pi}(\Theta-|f(\omega)|^{-2})^+\mathrm{d}\omega=2\Theta-\frac{2}{k+1}\frac{1}{\|\underline{r}\|_2^2}.
\end{eqnarray}
For $i=1,2$ define 
\begin{eqnarray}
I_i:=\frac{1}{2\pi}\int_0^{2\pi}(\Theta_i-|f(\omega)|^{-2})^+\mathrm{d}\omega=\left\{\begin{array}{cc}
   P   &  i=1  \\
   2\Theta_2-\frac{2}{k+1}\frac{1}{\|\underline{r}\|_2^2}   &   i=2
\end{array}\right.,
\end{eqnarray}
\begin{eqnarray}
d_{\min,i}:=(\Theta_i-\alpha^{-2})^+,\,\,\,\,\,d_{\max,i}:=(\Theta_i-\beta^{-2})^+,\,\,\,\,\,%\lambda_{\max}=\lim_{n\to\infty}\lambda_{\max}(H_c^TH_c)
\end{eqnarray}
\begin{eqnarray}
\label{eleven}
\delta_i:=-\frac{1}{2}\log\Big(1-\frac{r_s(r_s+2\beta)d_{\max,i}}{1+\alpha^2d_{\min,i}}\Big)+\frac{1}{2\ln2}\Big(1-\frac{(1-r_s(r_s+2\beta)I_i)^+}{1+r_s(r_s+2\beta)d_{\max,i}}\Big)
\end{eqnarray}
and
\begin{eqnarray}
\label{sang_po}
C^*_{LB,i}:=\frac{1}{4\pi}\int_0^{2\pi}\log[\max(\Theta_i|f(\omega)|^2,1)]\mathrm{d}\omega-\log\Big(1+\frac{k+1}{2}\|\underline{r}\|_2^2I_i\Big)-\delta_i.
\end{eqnarray} 

%and\,\footnote{. }
%and
%\begin{eqnarray}
%d_{\min}=(\Theta-\lambda^{-1}_{\min})^+,\,\,\,\,\,\lambda_{\min}=\lim_{n\to\infty}\lambda_{\min}(H_c^TH_c).
%\end{eqnarray}
 Then $C^*$ is bounded from below by $C^*_{LB,1}$. If $I_1\geq I_2$, then $C^*_{LB,2}$ is also a lower bound on $C^*$. 

\end{proposition}  
\begin{proof}
See Appendix~B.
\end{proof}
There is a slight abuse of notation here. The index $i$ in $\delta_i$ in (\ref{eleven}) only takes on $1,2$ and must not be mistaken with the index $n$ in $\delta_n$ in (\ref{pookl}). We see that if $r_0=r_1=\cdots=r_k=0$, the lower bound $C^*_{LB,1}$ reduces to the capacity $C^*_0$ in (\ref{known_1}). An immediate corollary of this proposition is an upper bound on the reduction in the channel capacity compared to when  all radii $r_i$ are zeros.  
\begin{coro}
We have 
\begin{eqnarray}
\label{omran}
C^*_0-C^*\leq \log\Big(1+\frac{k+1}{2}\|\underline{r}\|_2^2P\Big)+\delta_1,
\end{eqnarray}
where $\delta_1$ is given in (\ref{eleven}).
\end{coro}
\begin{proof}
Since $\Theta_1$ solves the equation in (\ref{faghat}), then $C^*_0=\frac{1}{4\pi}\int_0^{2\pi}\log[\max(\Theta_1|f(\omega)|^2,1)]\mathrm{d}\omega$. Hence, the lower bound $C^*_{LB,1}$ in (\ref{sang_po}) can be written as $C^*_{LB,1}=C^*_0-\log(1+\frac{k+1}{2}\|\underline{r}\|_2^2I_1)-\delta_1$. This verifies (\ref{omran}).
\end{proof}
We note that $\delta_1$ depends on the radii $r_i$ only through their sum $r_s$. Since $\|\underline{r}\|_2^2\leq (\sum_{i=0}^k r_i)^2=r^2_s$, one can loosen the upper bound on $C^*_0-C^*$ in (\ref{omran}) in order to write it entirely in terms of $r_s$ as  
\begin{eqnarray}
\label{pillow}
C^*_0-C^*&\leq& \log\Big(1+\frac{k+1}{2}r^2_sP\Big)-\frac{1}{2}\log\Big(1-\frac{r_s(r_s+2\beta)d_{\max,1}}{1+\alpha^2d_{\min,1}}\Big)\notag\\
&&\hskip1.5cm+\frac{1}{2\ln2}\Big(1-\frac{(1-r_s(r_s+2\beta)P)^+}{1+r_s(r_s+2\beta)d_{\max,1}}\Big).
\end{eqnarray} 
\begin{figure}[t]
\centering
  \includegraphics[scale=0.45]{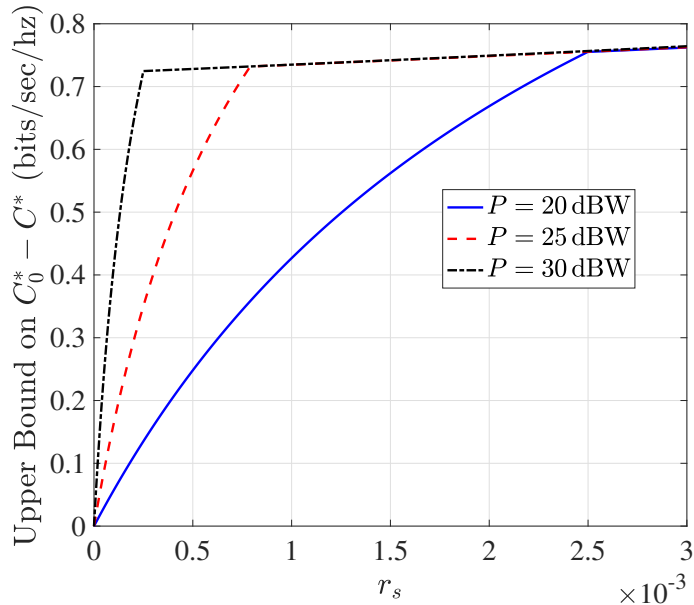}
  \caption{Plots of the upper bound in (\ref{pillow}) on the capacity loss $C_0^*-C^*$ for an ISI channel with $k=2$, $c_0=1$ and $c_1=c_2=0.5$. }
  \label{fig1}
\end{figure}
Figure~\ref{fig1} presents plots of the upper bound on $C^*_0-C^*$ given in (\ref{pillow}) in terms of $r_s$ for several values of $P$. The channel parameters are $k=2$, $c_0=1$, $c_1=c_2=0.5$. The third term on the right side of (\ref{pillow}) saturates at $\frac{1}{2\ln2}\approx0.7213$ bits/sec/hz as soon as $r_s$ is large enough such that $r_s(r_s+2\beta)P=1$. This value of $r_s$ reduces as $P$ increases. The first two terms on the right side of (\ref{pillow}) have a slower growth rate than the third term. This is why the given upper bound on $C^*_0-C^*$ stays around $\frac{1}{2\ln 2}$ for a while as $r_s$ keeps increasing. 

Next, we look into the solutions $\Theta_1$ and $\Theta_2$ to the equations (\ref{pre_ziba_1}) and (\ref{pre_ziba_2}), respectively. Let us define 
\begin{eqnarray}
J:=\frac{1}{2\pi}\int_0^{2\pi}|f(\omega)|^{-2}\mathrm{d}\omega. 
\end{eqnarray}
 It is easy to find $\Theta_1$ when $P$ is sufficiently large. To see this, assume $\Theta_1\geq\max_{0\leq \omega\leq 2\pi}|f(\omega)|^{-2}=\frac{1}{\alpha^2}$. Then $\frac{1}{2\pi}\int_0^{2\pi}(\Theta_1-|f(\omega)|^{-2})^+\mathrm{d}\omega=\Theta_1-J$ and (\ref{pre_ziba_1}) gives $\Theta_1=P+J$. If this value happens to be larger than or equal to $\frac{1}{\alpha^2}$,  it must be the unique solution to (\ref{pre_ziba_1}), i.e., 
\begin{eqnarray}
\label{set_wer}
 P\geq \frac{1}{\alpha^2}-J \implies \Theta_1=P+J.
\end{eqnarray}
Then the lower bound $C^*_{LB,1}$ is written as
\begin{eqnarray}
\label{lake}
C^*_{LB,1}=\frac{1}{2\pi}\int_0^{2\pi}\log|f(\omega)|\mathrm{d}\omega+\frac{1}{2}\log(P+J)-\log\Big(1+\frac{k+1}{2}\|\underline{r}\|_2^2P\Big)-\delta_1.
\end{eqnarray}
 It is also easy to find  $\Theta_2$ when $\|\underline{r}\|_2^2$ is sufficiently small. A similar argument that led to~(\ref{set_wer})~gives  
 \begin{eqnarray}
\frac{2}{k+1}\frac{1}{\|\underline{r}\|_2^2}\geq \frac{1}{\alpha^2}+J\implies \Theta_2=\frac{2}{k+1}\frac{1}{\|\underline{r}\|_2^2}-J.
\end{eqnarray}
Then the lower bound $C^*_{LB,2}$ is written as 
 \begin{eqnarray}
 \label{lake2}
C^*_{LB,2}=\frac{1}{2\pi}\int_0^{2\pi}\log|f(\omega)|\mathrm{d}\omega-\log((k+1)\|\underline{r}\|_2^2)-\frac{1}{2}\log\Big(\frac{2}{k+1}\frac{1}{\|\underline{r}\|_2^2}-J\Big)-\delta_2.
\end{eqnarray}
%\frac{1}{2}\log\Big(\frac{6}{\|\underline{r}\|_2^2}-J\Big)-\log\Big(2-\frac{\|\underline{r}\|_2^2}{3}J\Big)
The lower bound $C^*_{LB,1}$ reaches a maximum and eventually decreases as $P$ increases. The lower bound $C^*_{LB,2}$ does not depend on $P$. In fact, if one maximizes $\frac{1}{2}\log(P+J)-\log(1+\frac{k+1}{2}\|\underline{r}\|_2^2P)$ as a function of~$P$ in the given expression for $C^*_{LB,1}$ in (\ref{lake}), the maximum occurs at $P=\frac{2}{k+1}\frac{1}{\|\underline{r}\|_2^2}-2J$ and the resulting maximum value is precisely $-\log((k+1)\|\underline{r}\|_2^2)-\frac{1}{2}\log(\frac{2}{k+1}\frac{1}{\|\underline{r}\|_2^2}-J)$ which appears in the given expression for $C^*_{LB,2}$ in (\ref{lake2}). Moreover, recall that $C^*_{LB,2}$ is a lower bound on $C^*$ only when $I_1\ge I_2$. This inequality gives $P\geq 2(\frac{2}{k+1}\frac{1}{\|\underline{r}\|_2^2}-J)-\frac{2}{k+1}\frac{1}{\|\underline{r}\|_2^2}=\frac{2}{k+1}\frac{1}{\|\underline{r}\|_2^2}-2J$. Therefore, $\frac{2}{k+1}\frac{1}{\|\underline{r}\|_2^2}-2J$ is also the smallest value for $P$ beyond which $C^*_{LB,2}$ is a valid lower bound on $C^*$. We refer to this value of $P$ as the saturation power and denote it by $P_{sat}$, i.e., 
\begin{eqnarray}
\label{sat_p}
P_{sat}:=\frac{2}{k+1}\frac{1}{\|\underline{r}\|_2^2}-2J.
\end{eqnarray}  
\begin{figure}[t]
\centering
  \includegraphics[scale=0.5]{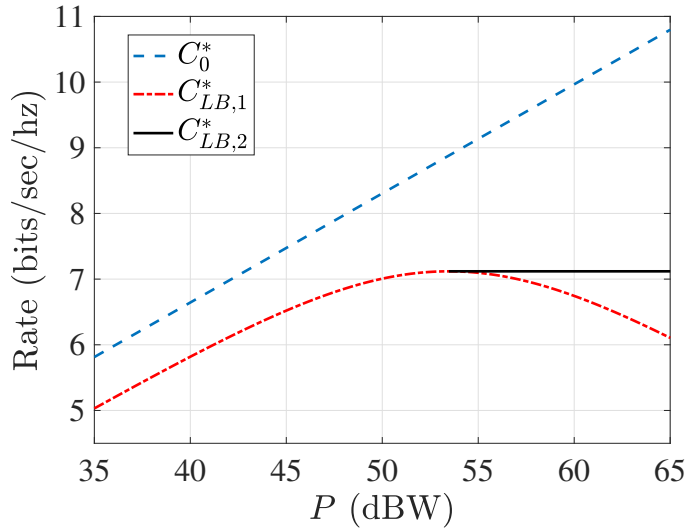}
  \caption{Plots of $C^*_0$ and the lower bounds $C^*_{LB,1}$ and $C^*_{LB,2}$ for an ISI channel with $k=2$, $c_0=1$, $c_1=c_2=0.5$ and $r_0=r_1=r_2=0.001$. The saturation power is $P_{sat}= 53.47$ dBW. }
  \label{fig2}
\end{figure}
As an example, consider an ISI channel with parameters $k=2$, $c_0=1$, $c_1=c_2=0.5$ and $r_0=r_1=r_2=0.001$. Then $P_{sat}\approx 53.47\,\mathrm{dBW}$. Figure~\ref{fig2} presents plots of the capacity $C^*_{0}$ and the lower bounds $C^*_{LB,1}$ and $C^*_{LB,2}$ in terms of $P$. The gap in between $C^*_0$ and $C^*_{LB,1}$ is less than $1$ bit/sec/hz for values of $P$ as large as $45\,\mathrm{dBW}$. In general, we have the next result which gives an upper bound on $C_0^*-C^*$ before saturation occurs. 
\begin{coro}
Assume $\frac{2}{k+1}\frac{1}{\|\underline{r}\|_2^2}\geq \frac{1}{\alpha^2}+J$. Then regardless of the value of $P\leq P_{sat}$, the capacity loss is bounded as 
\begin{eqnarray}
\label{omran_poi}
C_0^*-C^*\leq 1+\frac{1}{2\ln 2}-\frac{1}{2}\log\Big(1-\frac{r_s(r_s+2\beta)}{\alpha^2}\Big).
\end{eqnarray} 
\end{coro}
\begin{proof}
By (\ref{omran}) in Corollary~1, 
\begin{eqnarray}
\label{prod_1}
C^*_0-C^*&\leq& \log\Big(1+\frac{k+1}{2}\|\underline{r}\|_2^2P\Big)-\frac{1}{2}\log\Big(1-\frac{r_s(r_s+2\beta)d_{\max,1}}{1+\alpha^2d_{\min,1}}\Big)\notag\\
&&\hskip1.5cm+\frac{1}{2\ln2}\Big(1-\frac{(1-r_s(r_s+2\beta)P)^+}{1+r_s(r_s+2\beta)d_{\max,1}}\Big).
\end{eqnarray} 
Since $P\leq P_{sat}=\frac{2}{k+1}\frac{1}{\|\underline{r}\|_2^2}-2J$, then 
\begin{eqnarray}
\log\Big(1+\frac{k+1}{2}\|\underline{r}\|_2^2P\Big)\leq \log\Big(1+\frac{k+1}{2}\|\underline{r}\|_2^2P_{sat}\Big)=\log(2-(k+1)\|\underline{r}\|_2^2J)\leq \log 2=1.
\end{eqnarray}
The third term on the right side of (\ref{prod_1}) is never larger than $\frac{1}{2\ln 2}$. Thus,
\begin{eqnarray}
\label{perfect_1}
C_0^*-C^*\leq 1+\frac{1}{2\ln 2}-\frac{1}{2}\log\Big(1-\frac{r_s(r_s+2\beta)d_{\max,1}}{1+\alpha^2d_{\min,1}}\Big).
\end{eqnarray}
The condition $\frac{2}{k+1}\frac{1}{\|\underline{r}\|_2^2}\geq \frac{1}{\alpha^2}+J$ gives $P_{sat}\geq \frac{1}{\alpha^2}-J$. We study the cases $ \frac{1}{\alpha^2}-J\leq P\leq P_{sat}$ and $P\leq \frac{1}{\alpha^2}-J$~separately. 
\begin{enumerate}
  \item First, let $P\geq\frac{1}{\alpha^2}-J$.  Then $\Theta_1=P+J$ by (\ref{set_wer}) and we have $d_{\min,1}=(P+J-\alpha^{-2})^+=P+J-\alpha^{-2}$ and $d_{\max,1}=(P+J-\beta^{-2})^+=P+J-\beta^{-2}$. Hence, 
  \begin{eqnarray}
  \label{perfect_2}
\frac{d_{\max,1}}{1+\alpha^2d_{\min,1}}&=&\frac{P+J-\beta^{-2}}{1+\alpha^2(P+J-\alpha^{-2})}\notag\\
&=&\frac{P+J-\beta^{-2}}{\alpha^2(P+J)}\notag\\
&=&\frac{1}{\alpha^2}\Big(1-\frac{\beta^{-2}}{P+J}\Big)\leq\frac{1}{\alpha^2}.
\end{eqnarray}
Then (\ref{perfect_1}) and (\ref{perfect_2}) give the desired result.
  \item Next, let $P\leq \frac{1}{\alpha^2}-J$. We note that the solution $\Theta_1$ to (\ref{pre_ziba_1}) is nondecreasing in terms of $P$. Since the condition $P\geq \frac{1}{\alpha^2}-J$ gives $\Theta_1=P+J$, then the condition $P\leq \frac{1}{\alpha^2}-J$ implies $\Theta_1\leq P+J$. Hence, $d_{\min,1}=(\Theta_1-\alpha^{-2})^+\leq (P+J-\alpha^{-2})^+=0$ and $d_{\max,1}=(\Theta_1-\beta^{-2})^+\leq P+J-\beta^{-2}\leq \alpha^{-2}-\beta^{-2}$. Hence, once again the inequality in (\ref{perfect_2}) holds and the desired result follows. 
\end{enumerate}
\end{proof}
As an example, consider the channel with $k=2$, $c_0=1$, $c_1=c_2=0.5$ and $r_0=r_1=r_2=0.001$.  Then $r_s=0.003$, $\alpha=0.4677$, $\beta=2$ and Corollary~2 guarantees that $C_0^*-C^*\leq 1.7621$ bits/sec/hz regardless of $P\leq P_{sat}= 53.47\,\mathrm{dBW}$.

\section{The Decoder Structure}
%We obtain an upper bound on the average probability of error $\mathbb{E}[p_e(\boldsymbol{\mathcal{C}})]$ for the joint typicality decoder. 
The proposed decoder applies Gaussian joint-typicality decoding tuned to the matrix $H_c$ defined in~(\ref{vrty}). Denote the codewords by the independent random vectors $\underline{\boldsymbol{x}}_i\sim \mathrm{N}(\underline{0}_n,\Sigma_n)$ for $i=1,2,\cdots, 2^{nR}$ where $\Sigma_n$ is a positive-definite matrix. % such that 
% \begin{eqnarray}
%\mathrm{tr}(\Sigma_n)\leq nP.
%\end{eqnarray}
 Recall the definition of Gaussian typical sets in~(\ref{fgb_111}). The Gaussian joint-typicality decoder looks for the unique index $i=1,\cdots,2^{ nR}$ such that
\begin{eqnarray}
\label{e1}
\underline{\boldsymbol{x}}_i\in \mathcal{T}_{\varepsilon}^{(n)}(\Sigma_n),\,\,\,\underline{\boldsymbol{w}}_i\in \mathcal{T}_{\eta}^{(m+n)}(\Xi_n),\,\,\,
\end{eqnarray}
 where  $\varepsilon,\eta>0$ are constants and $\underline{\boldsymbol{w}}_i$ and $\Xi_n$ are defined by
\begin{eqnarray}
\label{wiii}
\underline{\boldsymbol{w}}_i:=\big[
    \underline{\boldsymbol{x}}^{T}_i\,\,\,\,\,  \underline{\boldsymbol{y}}^{\,T}\big]^T,\,\,\,\,\,\,\,\,\,\Xi_n:=\begin{bmatrix}
    \Sigma_n  & \Sigma_n H_c^T   \\
    H_c\Sigma_n  & I_m+H_c\Sigma_n H_c^T
\end{bmatrix}.
\end{eqnarray}
%and%\footnote{The matrix $\Sigma_n_c$ is the covariance matrix for $(\underline{\widetilde{\boldsymbol{x}}}^T\,\,\,\underline{\boldsymbol{\widetilde{y}}}^T)^T$ where $\underline{\boldsymbol{\widetilde{y}}}=H_c\underline{\boldsymbol{\widetilde{x}}}+\underline{\boldsymbol{\widetilde{z}}}$ and $\underline{\boldsymbol{\widetilde{x}}}$ and $\underline{\boldsymbol{\widetilde{z}}}$ are independent $\mathrm{N}(\underline{0}_n,PI_n)$ and $\mathrm{N}(\underline{0}_m,I_m)$ random vectors, respectively. }
%\begin{equation}
%\label{sigma1}
%
%\end{equation}
The matrix $\Xi_n$ is the covariance matrix for a vector $[\underline{\widetilde{\boldsymbol{x}}}^T\,\,\,\underline{\boldsymbol{\widetilde{y}}}^T]^T$ where $\underline{\boldsymbol{\widetilde{y}}}=H_c\,\underline{\boldsymbol{\widetilde{x}}}+\underline{\boldsymbol{\widetilde{z}}}$ and $\underline{\boldsymbol{\widetilde{x}}}$ and $\underline{\boldsymbol{\widetilde{z}}}$ are independent $\mathrm{N}(\underline{0}_n,\Sigma_n)$ and $\mathrm{N}(\underline{0}_m,I_m)$ random vectors, respectively. To construct the set $\mathcal{T}_\eta^{(m+n)}(\Xi_n)$, one needs to compute $\Xi_n^{-1}$. It is easy to see that $\Xi_n^{-1}$ admits a closed form given by 
\begin{eqnarray}
\label{sig_inv}
\Xi_n^{-1}=\begin{bmatrix}
    \Sigma_n^{-1}+H_c^TH_c & -H_c^T   \\
    -H_c & I_m 
\end{bmatrix}.
\end{eqnarray}
It will be useful in the course of our computations to note that %\footnote{This follows by the Schur's complement formula for calculating the determinant of partitioned matrices~\cite{Horn}.}
\begin{eqnarray}
\label{det_11}
\det(\Xi_n)=\det(\Sigma_n).
\end{eqnarray}
\section{Error Analysis}
Two types of error are distinguished: 
\begin{enumerate}[(i)]
  \item The transmitted codeword does not satisfy the decoding rule in~(\ref{e1}). We refer to this as the \textit{type~I~error} and denote it by $\mathcal{E}^{(I)}$.
  \item  A codeword different from the transmitted codeword satisfies (\ref{e1}). We refer to this as the \textit{type~II~error} and denote it by $\mathcal{E}^{(II)}$.
\end{enumerate}
   Then 
\begin{eqnarray}
\label{erri_1}
 p_{e,n}\leq \Pr(\mathcal{E}^{(I)})+\Pr(\mathcal{E}^{(II)}).
\end{eqnarray}
 In the following we examine the two terms on the right side of (\ref{erri_1}) separately.
\subsection{The probability of the type~I error}
Without loss of generality, assume $\underline{\boldsymbol{x}}_1$ is the transmitted codeword, i.e., $\underline{\boldsymbol{y}}=\boldsymbol{H}\underline{\boldsymbol{x}}_1+\underline{\boldsymbol{z}}$.  Then
\begin{eqnarray}
\label{ha_1}
\Pr(\mathcal{E}^{(I)})\leq\Pr\big(\underline{\boldsymbol{x}}_1\notin \mathcal{T}_{\varepsilon}^{(n)}(\Sigma_n)\big)+\Pr\big(\underline{\boldsymbol{w}}_1\notin \mathcal{T}_{\eta}^{(m+n)}(\Xi_n)\big),
\end{eqnarray}
 For every $\varepsilon>0$, the first term on the right side of (\ref{ha_1}) tends to zero when $n$ grows due to Theorem~5 in~\cite{Cover-Pombra}. The second term on the right side~(\ref{ha_1}) is studied in the next proposition: 
\begin{proposition}
\label{prop_1}
 Assume 
\begin{eqnarray}
\label{mir_111}
\eta>\eta_n:=(k+1)\|\underline{r}\|_2^2\,\frac{\mathrm{tr}(\Sigma_n)}{m+n}.%:=\varrho_0\big(P r_s^2,\sqrt{P}r\big),
\end{eqnarray} 
 Then 
  \begin{eqnarray}
 \label{gbm}
\Pr\big( \underline{\boldsymbol{w}}_1\notin \mathcal{T}_{\eta}^{(m+n)}(\Xi_n)\big)\leq C_n/((m+n)^2(\eta-\eta_n)^2),
\end{eqnarray}
where 
\begin{eqnarray}
\label{cons_C}
C_n:=2m+2n+8(k+1)nP\|\underline{r}\|_2^2+2nPr_s^4\lambda_{\max}(\Sigma_n).
\end{eqnarray}
\end{proposition}
\begin{proof}
%We have
 %\begin{eqnarray}
%\Pr(\underline{\boldsymbol{w}}_1\notin \mathcal{T}_{\eta}^{(m+n)}(\Sigma_n_c))%&=&\Pr\Big( \Big|\frac{1}{m+n}\underline{\boldsymbol{w}}_1^T \Sigma_n_c^{-1}\underline{\boldsymbol{w}}_1-1\Big|\ge\eta\Big)\notag\\
%\leq\Pr\Big(\frac{1}{m+n}\underline{\boldsymbol{w}}_1^{\,T}\Sigma_n_c^{-1}_{c}\underline{\boldsymbol{w}}_1\ge1+\eta\Big)+\Pr\Big(\frac{1}{m+n}\underline{\boldsymbol{w}}_1^{\,T}\Sigma_n_c^{-1}_{c}\underline{\boldsymbol{w}}_1\le1-\eta\Big).
%\end{eqnarray}
The event $\underline{\boldsymbol{w}}_1\notin \mathcal{T}_{\eta}^{(m+n)}(\Xi_n)$ is equivalent to $|\frac{1}{m+n}\underline{\boldsymbol{w}}_1^T \Xi_{n}^{-1}\underline{\boldsymbol{w}}_1-1|\geq\eta$. Let us study the term $\underline{\boldsymbol{w}}_1^T \Xi_{n}^{-1}\underline{\boldsymbol{w}}_1$. We write
\begin{eqnarray}
\label{hal_1}
\underline{\boldsymbol{w}}_1=\boldsymbol{A}\underline{\boldsymbol{\nu}},
\end{eqnarray}
where $\boldsymbol{A}$ and $\underline{\boldsymbol{\nu}}\sim \mathrm{N}(\underline{0}_{m+n},I_{m+n})$ are given by\footnote{Here, $\Sigma_n^{1/2}$ is a positive-definite matrix whose square is $\Sigma_n$. }
\begin{eqnarray}
\label{pure_1}
\boldsymbol{A}:=\begin{bmatrix}
   \Sigma_n^{1/2}   & 0_{n,m}   \\
   \boldsymbol{H}\Sigma_n^{1/2}  & I_m 
\end{bmatrix},\,\,\,\,\,\,\,\,\,\underline{\boldsymbol{\nu}}:=\begin{bmatrix}
     \Sigma_n^{-1/2}\underline{\boldsymbol{x}}^{T}_1 & \underline{\boldsymbol{z}}^T
\end{bmatrix}^T.
\end{eqnarray}
%and
%\begin{equation}
%\label{pure_657}
%\underline{\boldsymbol{\nu}}:=\begin{bmatrix}
 %    \frac{1}{\sqrt{P}}\underline{\boldsymbol{x}}_1^{\,T} & \underline{\boldsymbol{z}}^T
%\end{bmatrix}^T.
%\end{equation}
Then
\begin{eqnarray}
\label{hal_2}
\underline{\boldsymbol{w}}^T_1\Xi_n^{-1}\underline{\boldsymbol{w}}_1=\underline{\boldsymbol{\nu}}^T\boldsymbol{\Phi}\,\underline{\boldsymbol{\nu}},
\end{eqnarray}
 where $\boldsymbol{\Phi}$ is given by
\begin{eqnarray}
\label{b_b1}
\boldsymbol{\Phi}:=\boldsymbol{A}^T\Xi_n^{-1}\boldsymbol{A}.
\end{eqnarray}
Substituting the expression given for $\Xi_n^{-1}$ in (\ref{sig_inv}) and performing simple algebra, we find that
\begin{eqnarray}
\label{goli_111}
\boldsymbol{\Phi}=\begin{bmatrix}
   I_n+\Sigma_n^{1/2}\boldsymbol{E}^T\boldsymbol{E}\Sigma_n^{1/2}  & \Sigma_n^{1/2}\boldsymbol{E}^T  \\
\boldsymbol{E}     \Sigma_n^{1/2}&  I_{m}
\end{bmatrix},
\end{eqnarray}
where the ``\textit{error matrix}'' $\boldsymbol{E}$ is defined by
\begin{eqnarray}
\label{err_111}
\boldsymbol{E}:=\boldsymbol{H}-H_c,
\end{eqnarray}
i.e., $\boldsymbol{E}$ is the difference between the actual channel matrix $\boldsymbol{H}$ and $H_c$. By (\ref{hal_2})
 \begin{eqnarray}
 \label{behem_1}
\Pr\big(\underline{\boldsymbol{w}}_1\notin \mathcal{T}_{\eta}^{(m+n)}(\Xi_n )\big)=\Pr\Big(\Big|\frac{1}{m+n}\underline{\boldsymbol{\nu}}^{\,T}\boldsymbol{\Phi}\,\underline{\boldsymbol{\nu}}-1\Big|\ge\eta\,\Big).
\end{eqnarray}
  We continue to further bound the term on the right side of (\ref{behem_1}). Recall the range for the random matrix~$\boldsymbol{H}$ was denoted by $\mathcal{H}$ in Section~II. For $H\in \mathcal{H}$, we denote the corresponding realizations for $\boldsymbol{E}$ and $\boldsymbol{\Phi}$ by $E$ and $\Phi$, respectively.   We write
\begin{eqnarray}
\label{doos11_11_D}
&&\Pr\big(|\underline{\boldsymbol{\nu}}^{\,T}\boldsymbol{\Phi}\underline{\boldsymbol{\nu}}-(m+n)|\ge(m+n)\eta\big)\notag\\
&&\hskip1cm\stackrel{}{=}\int_{\mathcal{H}}\Pr\big(|\underline{\boldsymbol{\nu}}^{\,T}\boldsymbol{\Phi}\underline{\boldsymbol{\nu}}-(m+n)|\ge(m+n)\eta\,|\,\boldsymbol{H}=H\big)\mathrm{d}\mathcal{L}_{\boldsymbol{H}}(H),
\end{eqnarray}
where $\mathcal{L}_{\boldsymbol{H}}(\cdot)$ is the probability law for $\boldsymbol{H}$. We have
\begin{eqnarray}
\label{fgki}
&&\Pr\big(|\underline{\boldsymbol{\nu}}^{\,T}\boldsymbol{\Phi}\underline{\boldsymbol{\nu}}-(m+n)|\ge(m+n)\eta\,|\,\boldsymbol{H}=H\big)\notag\\
&&\hskip1cm\stackrel{(a)}{\leq}\Pr\big(|\underline{\boldsymbol{\nu}}^{\,T}\Phi\underline{\boldsymbol{\nu}}-(m+n)|\ge(m+n)\eta\big)\notag\\
&&\hskip1cm\stackrel{(b)}{\leq}\Pr\Big(|\underline{\boldsymbol{\nu}}^{\,T}\Phi\underline{\boldsymbol{\nu}}-\mathbb{E}[\underline{\boldsymbol{\nu}}^{\,T}\Phi\underline{\boldsymbol{\nu}}]|+|\mathbb{E}[\underline{\boldsymbol{\nu}}^{\,T}\Phi\underline{\boldsymbol{\nu}}]-(m+n)|\ge(m+n)\eta\Big)\notag\\
&&\hskip1cm\stackrel{}{=}\Pr\Big(|\underline{\boldsymbol{\nu}}^{\,T}\Phi\underline{\boldsymbol{\nu}}-\mathbb{E}[\underline{\boldsymbol{\nu}}^{\,T}\Phi\underline{\boldsymbol{\nu}}]|\ge(m+n)\big(\eta-|\mathbb{E}[\underline{\boldsymbol{\nu}}^{\,T}\Phi\underline{\boldsymbol{\nu}}]/(m+n)-1|\big)\Big)\notag\\
&&\hskip1cm\stackrel{(c)}{=}\Pr\Big(|\underline{\boldsymbol{\nu}}^{\,T}\Phi\underline{\boldsymbol{\nu}}-\mathbb{E}[\underline{\boldsymbol{\nu}}^{\,T}\Phi\underline{\boldsymbol{\nu}}]|\ge(m+n)\big(\eta-\mathrm{tr}(\Sigma_n^{1/2}E^TE\Sigma_n^{1/2})/(m+n)\big)\Big),
\end{eqnarray} 
where $(a)$ is due to independence of $\boldsymbol{H}$ and $\underline{\boldsymbol{\nu}}$, $(b)$ follows by adding and subtracting $\mathbb{E}[\underline{\boldsymbol{\nu}}^{\,T}\Phi\underline{\boldsymbol{\nu}}]$ and applying the triangle inequality and $(c)$ is due~to 
\begin{eqnarray}
\label{focus_11_D}
\mathbb{E}[\underline{\boldsymbol{\nu}}^{\,T}\Phi\underline{\boldsymbol{\nu}}]=\mathrm{tr}(\Phi)=m+n+\mathrm{tr}(\Sigma_n^{1/2}E^TE\Sigma_n^{1/2}).
\end{eqnarray}
 But, 
\begin{eqnarray}
\label{bill_111_D}
\mathrm{tr}(\Sigma_n^{1/2}E^TE\Sigma_n^{1/2})&\stackrel{(a)}{=}&\mathrm{tr}(E^TE\Sigma_n)\notag\\
&\stackrel{(b)}{=}&|\mathrm{tr}(E^TE\Sigma_n)|\notag\\
&=&\Big|\sum_{i=1}^n[E^TE\Sigma_n]_{i,i}\Big|\notag\\
&=&\Big|\sum_{i=1}^n\sum_{j=1}^n[E^TE]_{i,j}[\Sigma_n]_{j,i}\Big|\notag\\
&\stackrel{(c)}{\leq}&\sum_{i=1}^n\sum_{j=1}^n\big|[E^TE]_{i,j}\big|\big|\big[\Sigma_n]_{j,i}\big|\notag\\
&\stackrel{(d)}{\leq}&\sum_{i=1}^n\sum_{j=1}^n\big|[E^TE]_{i,j}\big|\sqrt{[\Sigma_n]_{i,i}[\Sigma_n]_{j,j}},
\end{eqnarray}
where $(a)$ is due to the identity $\mathrm{tr}(M_1M_2)=\mathrm{tr}(M_2M_1)$, $(b)$ is due to $\mathrm{tr}(E^TE\Sigma_n)\geq 0$, $(c)$ is due to the triangle inequality and $(d)$  is due to the fact that $\Sigma_n$ is a positive definite matrix and hence, each of its central $2\times 2$ minors are nonnegative, i.e., $[\Sigma_n]_{i,i}[\Sigma_n]_{j,j}-[\Sigma_n]^2_{j,i}\geq 0$. Moreover, 
\begin{eqnarray}
\label{tir_14}
\big|[E^TE]_{i,j}\big|&=&\Big|\sum_{l=1}^mE_{l,i}E_{l,j}\Big|\notag\\&\stackrel{(a)}{\leq}& \sum_{l=1}^m|E_{l,i}||E_{l,j}|\notag\\&\stackrel{(b)}{=}&\sum_{\substack{l=1\\0\leq l-i,l-j\leq k}}^m|E_{l,i}||E_{l,j}|\notag\\&\stackrel{(c)}{\leq}& \sum_{\substack{l=1\\0\leq l-i,l-j\leq k}}^mr_{l-i}r_{l-j},
\end{eqnarray}
where $(a)$ is due to the triangle inequality, $(b)$ is due to $E_{l,i}=0$ for $l-i<0$ and $l-i>k$ and $(c)$ is due to $|E_{l,i}|\leq r_{l-i}$ for $0\leq l-i\leq k$. By (\ref{bill_111_D}) and (\ref{tir_14}),
\begin{eqnarray}
\label{bill_111_D_111}
%\mathrm{tr}(\Sigma_n^{1/2}E^TE\Sigma_n^{1/2})&\leq& \sum_{i=1}^n\sum_{j=1}^n\sum_{\substack{l=1\\0\leq l-i,l-j\leq k}}^mr_{l-i}r_{l-j}\sqrt{[\Sigma_n]_{i,i}[\Sigma_n]_{j,j}}\notag\\
%&\stackrel{(a)}{=}&\sum_{l=1}^m\sum_{\substack{i=1\\ 0\leq l-i\leq k}}^n\,\sum_{\substack{j=1\\ 0\leq l-j\leq k}}^nr_{l-i}r_{l-j}\sqrt{[\Sigma_n]_{i,i}[\Sigma_n]_{j,j}}\notag\\
%&=&\sum_{l=1}^m\Big(\sum_{\substack{i=1\\ 0\leq l-i\leq k}}^nr_{l-i}\sqrt{[\Sigma_n]_{i,i}}\,\Big)^2\notag\\
%&\stackrel{(b)}{\leq}&\sum_{l=1}^m(k+1)\sum_{\substack{i=1\\ 0\leq l-i\leq k}}^nr_{l-i}^2[\Sigma_n]_{i,i}\notag\\
%&\stackrel{(c)}{= }&(k+1)\sum_{i=1}^n[\Sigma_n]_{i,i}\sum_{\substack{l=1\\ 0\leq l-i\leq k}}^mr_{l-i}^2\notag\\
%&\stackrel{(d)}{\leq}&(k+1)\|\underline{r}\|_2^2\sum_{i=1}^n[\Sigma_n]_{i,i}\notag\\
%&=&(k+1)\|\underline{r}\|_2^2\,\mathrm{tr}(\Sigma_n),
\mathrm{tr}(\Sigma_n^{1/2}E^TE\Sigma_n^{1/2})&\leq& \sum_{i=1}^n\sum_{j=1}^n\sum_{\substack{l=1\\0\leq l-i,l-j\leq k}}^mr_{l-i}r_{l-j}\sqrt{[\Sigma_n]_{i,i}[\Sigma_n]_{j,j}}\notag\\
&\stackrel{(a)}{=}&\sum_{l=1}^m\sum_{\substack{i=1\\ 0\leq l-i\leq k}}^n\,\sum_{\substack{j=1\\ 0\leq l-j\leq k}}^nr_{l-i}r_{l-j}\sqrt{[\Sigma_n]_{i,i}[\Sigma_n]_{j,j}}\notag\\
&=&\sum_{l=1}^m\Big(\sum_{\substack{i=1\\ 0\leq l-i\leq k}}^nr_{l-i}\sqrt{[\Sigma_n]_{i,i}}\,\Big)^2\notag\\
&\stackrel{(b)}{\leq}&\sum_{l=1}^m\Big(\sum_{\substack{i=1\\ 0\leq l-i\leq k}}^nr_{l-i}^2\Big)\Big(\sum_{\substack{i=1\\ 0\leq l-i\leq k}}^n[\Sigma_n]_{i,i}\Big)\notag\\
&\stackrel{(c)}{\leq }&\|\underline{r}\|_2^2\sum_{l=1}^m\sum_{\substack{i=1\\ 0\leq l-i\leq k}}^n[\Sigma_n]_{i,i}\notag\\
&\stackrel{(d)}{= }&\|\underline{r}\|_2^2\sum_{i=1}^n[\Sigma_n]_{i,i}\sum_{\substack{l=1\\ 0\leq l-i\leq k}}^m1\notag\\
&\stackrel{(e)}{\leq}&(k+1)\|\underline{r}\|_2^2\sum_{i=1}^n[\Sigma_n]_{i,i}\notag\\
&=&(k+1)\|\underline{r}\|_2^2\,\mathrm{tr}(\Sigma_n),
\end{eqnarray}
where in $(a)$ and $(d)$ we have changed the order of summations, $(b)$ is due to Cauchy-Schwarz  inequality, %\footnote{Alternatively, one can apply the Cauchy-Schwarz inequality in step (b)  in (\ref{bill_111_D_111}). At the end, it will result in the same upper bound in~(\ref{bill_111_D_111}). } 
 $(c)$ is due to $\sum_{\substack{i=1\\ 0\leq l-i\leq k}}^nr_{l-i}^2\leq \|\underline{r}\|_2^2$ and $(e)$ is due to $\sum_{\substack{l=1\\ 0\leq l-i\leq k}}^m1\leq k+1$.
%\begin{eqnarray}
%\label{rel_1}
%\mathbb{E}\big[[E^TE]_{i,i}\big]=\mathbb{E}\Big[\sum_{l=1}^mE_{l,i}^2\Big]\stackrel{(a)}{=}\mathbb{E}\Big[\sum_{l=i}^{i+k}E_{l,i}^2\Big]\stackrel{(b)}{=}\sum_{l=i}^{i+k}\frac{r^2_{l-i}}{3}=\frac{1}{3}\|\underline{r}\|_2^2,
%\end{eqnarray}
%where $(a)$ is due to $[E]_{l,i}$ being zero for $l<i$ or $l>i+k$ and $(b)$ is due to $[E]_{l,i}$ being uniformly distributed over $[-r_{l-i},r_{l-i}]$ for $i\le l\leq i+k$. By (\ref{bill_111_D}) and (\ref{rel_1}),  
%\begin{eqnarray}
%\label{zir_1}
%\mathbb{E}[\mathrm{tr}(\Sigma_n^{1/2}E^TE\Sigma_n^{1/2})]
%\stackrel{}{=}\frac{1}{3}\|\underline{r}\|_2^2\sum_{i=1}^n[\Sigma_n]_{i,i}
%\stackrel{}{=}\frac{1}{3}\|\underline{r}\|_2^2\mathrm{tr}(\Sigma_n).
%\end{eqnarray}
 By~(\ref{fgki}) and (\ref{bill_111_D_111}) and the definition of $\eta_n$ in (\ref{mir_111}),   
\begin{eqnarray}
\label{paolo}
&&\Pr\big(|\underline{\boldsymbol{\nu}}^{\,T}\boldsymbol{\Phi}\underline{\boldsymbol{\nu}}-(m+n)|\ge(m+n)\eta\,|\,\boldsymbol{H}=H\big)\notag\\
&&\hskip1cm\leq \Pr\big(|\underline{\boldsymbol{\nu}}^{\,T}\Phi\underline{\boldsymbol{\nu}}-\mathbb{E}[\underline{\boldsymbol{\nu}}^{\,T}\Phi\underline{\boldsymbol{\nu}}]|\ge(m+n)\big(\eta-\eta_n\big)\big)\notag\\
&&\hskip1cm\stackrel{(a)}{\leq} \frac{\mathrm{Var}(\boldsymbol{\underline{\nu}}^T\Phi\underline{\boldsymbol{\nu}})}{(m+n)^2(\eta-\eta_n)^2}\notag\\
&&\hskip1cm\stackrel{(b)}{=} \frac{2\mathrm{tr}(\Phi^2)}{(m+n)^2(\eta-\eta_n)^2},
\end{eqnarray}
where $(a)$ is due to $\eta>\eta_n$ and Chebyshev's inequality and $(b)$ is due to $\mathrm{Var}(\boldsymbol{\underline{\nu}}^T\Phi\underline{\boldsymbol{\nu}})=2\mathrm{tr}(\Phi^2)$.  It is shown in Appendix~C that 
\begin{eqnarray}
\label{c_n}
\mathrm{tr}(\Phi^2)\leq C_n/2,
\end{eqnarray}
 where $C_n$ is given in (\ref{cons_C}). Therefore, 
 \begin{eqnarray}
 \label{brad_pit}
\Pr\big(|\underline{\boldsymbol{\nu}}^{\,T}\boldsymbol{\Phi}\underline{\boldsymbol{\nu}}-(m+n)|\ge(m+n)\eta\,|\,\boldsymbol{H}=H\big)\leq \frac{C_n}{(m+n)^2(\eta-\eta_n)^2}.
\end{eqnarray}
The right side of (\ref{brad_pit}) does not depend on the choice of $H\in \mathcal{H}$. By (\ref{behem_1}), (\ref{doos11_11_D}) and (\ref{brad_pit}), we arrive at the promised bound in (\ref{gbm}).  
\end{proof}

By the previous proposition, we see that if $\lim_{n\to\infty}\frac{1}{n}\lambda_{\max}(\Sigma_n)=0$ and $\eta>\limsup_{n\to\infty}\eta_n$, then 
\begin{eqnarray}
\lim_{n\to\infty}\Pr\big(\underline{\boldsymbol{w}}_1\notin \mathcal{T}_{\eta}^{(m+n)}(\Xi_n)\big)=0.
\end{eqnarray}
%\begin{eqnarray}
%\label{nasa_1}
%\Pr\big(\underline{\boldsymbol{w}}(1)\notin \mathcal{T}_{\eta}^{(m+n)}(\Sigma_n_c)\big)\leq \frac{\mathrm{Var}(\boldsymbol{\underline{\nu}}^T\boldsymbol{\Phi}\underline{\boldsymbol{\nu}})}{(m+n)^2(\eta-\Xi_n)^2}.%=\frac{2\mathbb{E}[\mathrm{tr}(\boldsymbol{\Phi}^2)]}{(m+n)^2(\eta-\Xi_n)^2},
%\end{eqnarray}
\subsection{The probability of the type~II error}
Let $\underline{\boldsymbol{x}}_1$ be the transmitted codeword. Define
%\begin{eqnarray}
%\Omega_c:=I_m+H_c\Sigma_n H_c^T%
%\end{eqnarray} 
\begin{eqnarray}
\label{omega_c}
\Omega_c:=I_m+H_c\Sigma_n H_c^T,
\end{eqnarray}
where $H_c$ is defined in (\ref{vrty}).
  Let $\eta'>0$ to be determined. The probability of the type~II error can be bounded as  
\begin{eqnarray}
\label{pish_11_11}
\Pr(\mathcal{E}^{(II)})&=&\Pr\big(\exists\,i\neq 1,\,\,\,  \underline{\boldsymbol{x}}_i\in \mathcal{T}_{\varepsilon}^{(n)}(\Sigma_n),\underline{\boldsymbol{w}}_i\in  \mathcal{T}_{\eta}^{(m+n)}(\Xi_n)\big)\notag\\&\le&\Pr\big(\underline{\boldsymbol{y}}\notin \mathcal{T}_{\eta'}^{(m)}(\Omega_c)\big)\notag\\
&&\hskip1cm+\Pr\big(\exists\,i\neq 1,\,\,\,\underline{\boldsymbol{x}}_i\in \mathcal{T}_{\varepsilon}^{(n)}(\Sigma_n), \underline{\boldsymbol{w}}_i\in \mathcal{T}_{\eta}^{(m+n)}(\Xi_n),  \underline{\boldsymbol{y}}\in \mathcal{T}_{\eta'}^{(m)}(\Omega_c) \big)\notag\\
&\le&\Pr\big(\underline{\boldsymbol{y}}\notin \mathcal{T}_{\eta'}^{(m)}(\Omega_c)\big)\notag\\
&&\hskip1cm+\sum_{i=2}^{ 2^{nR}}\Pr\big( \underline{\boldsymbol{x}}_i\in \mathcal{T}_{\varepsilon}^{(n)}(\Sigma_n), \underline{\boldsymbol{w}}_i\in \mathcal{T}_{\eta}^{(m+n)}(\Xi_n),  \underline{\boldsymbol{y}}\in \mathcal{T}_{\eta'}^{(m)}(\Omega_c) \big),
\end{eqnarray} 
where we have applied the union bound in the last step. The next proposition studies the first term on the right side of~(\ref{pish_11_11}):
\begin{proposition}
\label{prop_2}
Assume 
\begin{eqnarray}
\label{berf_111}
%r\eta'>\xi'_n:=\varrho_0\big(Pr^2+2Pr\varrho_n,\sqrt{P}r\big),
\eta'>\eta'_n:=r_s(r_s+2\beta)\frac{\mathrm{tr}(\Sigma_n)}{m}.
\end{eqnarray} 
 Then 
\begin{eqnarray}
\label{peter_111}
\Pr\big(\underline{\boldsymbol{y}}\notin \mathcal{T}_{\eta'}^{(m)}(\Omega_c)\big)\leq C_n'/(m^2(\eta'-\eta'_n)^2),
\end{eqnarray}
where 
\begin{eqnarray}
\label{cons_CP}
C'_n:=2m+4(\beta+r_s)^2nP+2nP(\beta+r_s)^4\lambda_{\max}(\Sigma_n).
\end{eqnarray}
\end{proposition}
\begin{proof} 
Let us write
\begin{eqnarray}
\underline{\boldsymbol{y}}=\boldsymbol{B}\underline{\boldsymbol{\nu}},
\end{eqnarray}
where 
\begin{eqnarray}
\boldsymbol{B}:=\big[
     \boldsymbol{H}\Sigma_n^{1/2}\,\,\, I_m\big]
\end{eqnarray}
 and $\boldsymbol{\underline{\nu}}$ is given in (\ref{pure_1}).
Then 
\begin{eqnarray}
\label{al_1}
\underline{\boldsymbol{y}}^T\Omega_c^{-1}\underline{\boldsymbol{y}}=\underline{\boldsymbol{\nu}}^T\boldsymbol{\Psi}\boldsymbol{\underline{\nu}},
\end{eqnarray} 
where $\boldsymbol{\Psi}$ is defined by
\begin{eqnarray}
\label{psi_09}
\boldsymbol{\Psi}:=\boldsymbol{B}^T\Omega_c^{-1}\boldsymbol{B}=\begin{bmatrix}
      \Sigma_n^{1/2}\boldsymbol{H}^T \Omega_c^{-1}\boldsymbol{H}\Sigma_n^{1/2}  & \Sigma_n^{1/2}\boldsymbol{H}^T \Omega_c^{-1} \\
        \Omega_c^{-1}\boldsymbol{H}\Sigma_n^{1/2} & \Omega_c^{-1}
\end{bmatrix}.
\end{eqnarray}
By (\ref{al_1}), 
\begin{eqnarray}
\label{yellp_1}
\Pr\big(\underline{\boldsymbol{y}}\notin \mathcal{T}_{\eta'}^{(m)}(\Omega_c)\big)&=&\Pr\Big(\Big|\frac{1}{m}\underline{\boldsymbol{y}}^{\,T}\Omega_c^{-1}\,\underline{\boldsymbol{y}}-1\Big|\ge\eta'\,\Big)\notag\\&=&\Pr\big(|\underline{\boldsymbol{\nu}}^T\boldsymbol{\Psi}\boldsymbol{\underline{\nu}}-m|\geq m\eta'\big).
\end{eqnarray}
%Recall the error matrix $E$ in (\ref{err_111}). Writing $H=E+H_c$, we decompose $\boldsymbol{\Psi}$ as 
%\begin{eqnarray}
%\label{al_2}
%\boldsymbol{\Psi}=\boldsymbol{\Psi}_0+r',
%\end{eqnarray}
%where
%\begin{eqnarray}
%\boldsymbol{\Psi}_0:=\begin{bmatrix}
 %     PH_c^T \Omega_c^{-1}H_c  & \sqrt{P}H_c^T \Omega_c^{-1} \\
  %      \sqrt{P}\Omega_c^{-1}H_c & \Omega_c^{-1}
%\end{bmatrix}
%\end{eqnarray}
%and 
%\begin{eqnarray}
%\label{road_1}
%\Delta':=\begin{bmatrix}
   %   PE^T \Omega_c^{-1}E+PE^T \Omega_c^{-1}H_c+PH_c^T \Omega_c^{-1}E & \sqrt{P}E^T \Omega_c^{-1} \\
  %      \sqrt{P}\Omega_c^{-1}E & 0_{m,m}
%\end{bmatrix}
%\end{eqnarray}
%Let us write 
%\begin{eqnarray}
%\label{mis_11}
%\Pr\big(\underline{\boldsymbol{y}}\notin \mathcal{T}_{\eta'}^{(m)}(\Omega_c)\big)\le\Pr\Big(\frac{1}{m}\boldsymbol{\underline{y}}\Omega_c^{-1}\boldsymbol{\underline{y}}>1+\eta'\Big)+\Pr\Big(\frac{1}{m}\boldsymbol{\underline{y}}\Omega_c^{-1}\boldsymbol{\underline{y}}<1-\eta'\Big).
%\end{eqnarray}
Recall the range for the random matrix~$\boldsymbol{H}$ was denoted by $\mathcal{H}$ in Section~II. For $H\in \mathcal{H}$, we denote the corresponding realizations for $\boldsymbol{E}$ (the error matrix defined in (\ref{err_111})) and $\boldsymbol{\Psi}$ by $E$ and $\Psi$, respectively.   We write
\begin{eqnarray}
\label{doos11_11_D_11}
\Pr\big(|\underline{\boldsymbol{\nu}}^{\,T}\boldsymbol{\Psi}\underline{\boldsymbol{\nu}}-m|\ge m\eta'\big)\stackrel{}{=}\int_{\mathcal{H}}\Pr\big(|\underline{\boldsymbol{\nu}}^{\,T}\boldsymbol{\Psi}\underline{\boldsymbol{\nu}}-m|\ge m\eta'\,|\,\boldsymbol{H}=H\big)\mathrm{d}\mathcal{L}_{\boldsymbol{H}}(H),
\end{eqnarray}
where $\mathcal{L}_{\boldsymbol{H}}(\cdot)$ is the probability law for $\boldsymbol{H}$. Following similar lines of reasoning as in (\ref{fgki}),
\begin{eqnarray}
\label{alot_11_DD}
\Pr\big(|\underline{\boldsymbol{\nu}}^{\,T}\boldsymbol{\Psi}\underline{\boldsymbol{\nu}}-m|\ge m\eta'\,|\,\boldsymbol{H}=H\big)= \Pr\Big(|\underline{\boldsymbol{\nu}}^{\,T}\Psi\underline{\boldsymbol{\nu}}-\mathbb{E}[\underline{\boldsymbol{\nu}}^{\,T}\Psi\underline{\boldsymbol{\nu}}]|\ge m\big(\eta'-|\mathbb{E}[\underline{\boldsymbol{\nu}}^{\,T}\Psi\underline{\boldsymbol{\nu}}]/m-1|\big)\Big).
\end{eqnarray}
We note that
\begin{eqnarray}
\label{bolk_11_DD}
\mathbb{E}[\boldsymbol{\underline{\nu}}^T\Psi\boldsymbol{\underline{\nu}}]&\stackrel{}{=}&\mathrm{tr}(\Psi)\notag\\
%&=&\mathbb{E}[\mathrm{tr}(\Omega_c^{-1}+P\boldsymbol{H}^T\Omega_c^{-1}H)]\notag\\
%&=&\mathbb{E}[\mathrm{tr}(\Omega_c^{-1}+\Omega_c^{-1}PHH^T)]\notag\\
&\stackrel{(a)}{=}&\mathrm{tr}(\Omega_c^{-1}+\Sigma_n^{1/2}H^T \Omega_c^{-1}H\Sigma_n^{1/2} )\notag\\
&\stackrel{(b)}{=}&\mathrm{tr}(\Omega_c^{-1}+ \Omega_c^{-1}H\Sigma_nH^T)\notag\\
&=&\mathrm{tr}\big(\Omega_c^{-1}(I_m+H\Sigma_nH^T )\big )\notag\\
&\stackrel{(c)}{=}&\mathrm{tr}\big(\Omega_c^{-1}(\underbrace{I_m+H_c\Sigma_n H_c^T}_{=\Omega_c}+E\Sigma_n H_c^T+H_c\Sigma_nE^T+E\Sigma_nE^T)\big)\notag\\
&\stackrel{(d)}{=}&m+2\mathrm{tr}(\Omega_c^{-1}H_c\Sigma_nE^T)+\mathrm{tr}(\Omega_c^{-1}E\Sigma_nE^T),
\end{eqnarray}
where $(a)$ follows by the expression of $\Psi$ in (\ref{psi_09}), $(b)$ uses the identity $\mathrm{tr}(M_1M_2)=\mathrm{tr}(M_2M_1)$ twice to write $\mathrm{tr}(\Sigma_n^{1/2}H^T \Omega_c^{-1}H\Sigma_n^{1/2})=\mathrm{tr}(\Omega_c^{-1}H\Sigma_nH^T)$,  the error matrix $E$ in $(c)$ is defined in (\ref{err_111}) and $(d)$ uses the same trace identity used in $(b)$ together with $\mathrm{tr}(M)=\mathrm{tr}(M^T)$.\footnote{We have $\mathrm{tr}(\Omega_c^{-1}E\Sigma_nH_c^T)=\mathrm{tr}((\Omega_c^{-1}E\Sigma_nH_c^T)^T)=\mathrm{tr}((H_c\Sigma_nE^T\Omega_c^{-1})=\mathrm{tr}(\Omega_c^{-1}H_c\Sigma_nE^T)$.}
% By (\ref{alot_11_DD}) and (\ref{bolk_11_DD}), 
%\begin{eqnarray}
%\label{nasa_00}
%\Pr\big(\underline{\boldsymbol{y}}\notin \mathcal{T}_{\eta'}^{(m)}(\Omega_c)\big)\leq\Pr\Big(|\underline{\boldsymbol{\nu}}^{\,T}\boldsymbol{\Psi}\underline{\boldsymbol{\nu}}-\mathbb{E}[\underline{\boldsymbol{\nu}}^{\,T}\boldsymbol{\Psi}\underline{\boldsymbol{\nu}}]|\ge m\Big(\eta'-\frac{2}{m}|\mathbb{E}[\mathrm{tr}(\Omega_c^{-1}H_c\boldsymbol{E}^T)]|-\frac{1}{m}\mathbb{E}[\mathrm{tr}(\Omega_c^{-1}\boldsymbol{E}\Sigma_n\boldsymbol{E}^T)],\Big)\Big).
%&\leq&\Pr\big(|\underline{\boldsymbol{\nu}}^{\,T}\boldsymbol{\Psi}\underline{\boldsymbol{\nu}}-\mathbb{E}[\underline{\boldsymbol{\nu}}^{\,T}\boldsymbol{\Psi}\underline{\boldsymbol{\nu}}]|\ge m\big(\eta'-2P\mathbb{E}[|\mathrm{tr}(\Omega_c^{-1}H_c\boldsymbol{E}^T)|]/m\big)\big),
%\end{eqnarray}
Then 
\begin{eqnarray}
\label{ashegh_1}
|\mathbb{E}[\underline{\boldsymbol{\nu}}^{\,T}\Psi\underline{\boldsymbol{\nu}}]/m-1|&=&\big|2\mathrm{tr}(\Omega_c^{-1}H_c\Sigma_nE^T)+\mathrm{tr}(\Omega_c^{-1}E\Sigma_nE^T)\big|\notag\\
&\stackrel{}{\leq}&2\big|\mathrm{tr}(\Omega_c^{-1}H_c\Sigma_nE^T)\big|+\big|\mathrm{tr}(\Omega_c^{-1}E\Sigma_nE^T)\big|,
\end{eqnarray}
where we have used the triangle inequality. Next we compute upper bounds on the two terms on the right side~of~(\ref{ashegh_1}). We need the following key lemma: 
\begin{lem}
\label{lem_1}
Let $M_1$ and $M_2$ be matrices of sizes $m\times n$ and $n\times m$, respectively. Then 
\begin{eqnarray}
\label{yadam_1}
\|M_1M_2\|_2\leq \|M_1\|\|M_2\|_2
\end{eqnarray}
and 
\begin{eqnarray}
\label{yadam_2}
\|M_1M_2\|_2\leq \|M_2\|\|M_1\|_2
\end{eqnarray}
\begin{proof}
This is Problem 5.6.P20 in \cite{Horn}. A proof is provided in Appendix~D for completeness.  
\end{proof}
\end{lem} We have the thread of inequalities
\begin{eqnarray}
\label{babel_1}
|\mathrm{tr}(\Omega_c^{-1}H_c\Sigma_nE^T)|&=&|\mathrm{tr}(\Omega_c^{-1}H_c\Sigma^{1/2}_n\,\Sigma_n^{1/2}E^T)|\notag\\
&\stackrel{(a)}{\le}&\|\Omega_c^{-1}H_c\Sigma_n^{1/2}\|_2\|\Sigma_n^{1/2}E\|_2\notag\\
&\stackrel{(b)}{\le}&\|\Omega_c^{-1}H_c\|\|\Sigma_n^{1/2}\|_2\|E\|\|\Sigma_n^{1/2}\|_2\notag\\
&=&\|\Omega_c^{-1}H_c\|\,\|E\|\,\|\Sigma_n^{1/2}\|^2_2\notag\\
&\stackrel{(c)}{\leq}&\|\Omega_c^{-1}\|\|H_c\|\,\|E\|\,\|\Sigma_n^{1/2}\|^2_2\notag\\
&\stackrel{(d)}{\leq}&r_s \beta\mathrm{tr}(\Sigma_n),
%&\stackrel{(e)}{=}&\frac{1}{\sqrt{3}}\|\Omega_c^{-1}H_c\|_2\sqrt{n}\|\underline{r}\|_2,
%&\stackrel{(e)}{\leq}&\Big(\frac{1}{3}mn\sum_{k=0}^{k} r_{k}^2\Big)^{\frac{1}{2}}\varrho_n
\end{eqnarray}
where $(a)$ is due to Cauchy-Schwarz inequality $|\mathrm{tr}(XY^T)|\leq \|X\|_2\|Y\|_2$, $(b)$ applies Lemma~\ref{lem_1} twice, $(c)$ is due to $\|\cdot\|$ being sub-multiplicative and $(d)$ holds thanks to $\|\Sigma_n^{1/2}\|_2^2=\mathrm{tr}(\Sigma_n)$ and the three inequalities 
\begin{eqnarray}
\label{mahtab}
\|\Omega_c^{-1}\|\leq 1,\,\,\,\,\,\|H_c\|\leq \beta,\,\,\,\,\,\|E\|\leq r_s,
\end{eqnarray} 
where $r_s$ and $\beta$ are defined in (\ref{ehj_1}) and (\ref{min_max}), respectively. The first inequality $\|\Omega_c^{-1}\|\leq 1$ is due to $\|\Omega_c^{-1}\|^2=\lambda_{\max}(\Omega_c^{-2})$ and the fact that none of eigenvalues of $\Omega_c^{-2}=(I_m+H_c\Sigma_nH_c^T)^{-2}$ are larger than~$1$, the second inequality $\|H_c\|\leq \beta$
 is proved in Appendix~E and the third inequality $\|E\|\leq r_s$ is proved in Appendix~F.  Similarly,
\begin{eqnarray}
\label{babel_11}
\big|\mathrm{tr}(\Omega_c^{-1}E\Sigma_nE^T)\big|&\leq &\|\Omega_c^{-1}E\Sigma_n^{1/2}\|_2\|\Sigma_n^{1/2}E^T\|_2\notag\\
&\leq &\|\Omega_c^{-1}E\|\|\Sigma_n^{1/2}\|_2\|E^T\|\|\Sigma_n^{1/2}\|_2\notag\\
%&=&\mathbb{E}[\|\Omega_c^{-1}E\|\|E^T\|]\|\Sigma_n^{1/2}\|^2_2\notag\\
&\leq&\|\Omega_c^{-1}\|\|E\|\|E^T\|\|\Sigma_n^{1/2}\|^2_2\notag\\
&=&\|\Omega_c^{-1}\|\,\|E\|^2\|\Sigma_n^{1/2}\|^2_2\notag\\
&\leq& r_s^2\mathrm{tr}(\Sigma_n).
\end{eqnarray} 
Now, (\ref{alot_11_DD}), (\ref{ashegh_1}),  (\ref{babel_1}) and (\ref{babel_11}) yield
\begin{eqnarray}
\Pr\big(|\underline{\boldsymbol{\nu}}^{\,T}\boldsymbol{\Psi}\underline{\boldsymbol{\nu}}-m|\ge m\eta'\,|\,\boldsymbol{H}=H\big)&\leq&\Pr\big(|\underline{\boldsymbol{\nu}}^{\,T}\Psi\underline{\boldsymbol{\nu}}-\mathbb{E}[\underline{\boldsymbol{\nu}}^{\,T}\Psi\underline{\boldsymbol{\nu}}]|\ge m(\eta'-\eta'_n)\big)\notag\\
&\stackrel{(a)}{\leq}&\frac{\mathrm{Var}(\underline{\boldsymbol{\nu}}^{\,T}\Psi\underline{\boldsymbol{\nu}})}{m^2(\eta'-\eta'_n)^2}\notag\\
&\stackrel{(b)}{=}&\frac{2\mathrm{tr}(\Psi^2)}{m^2(\eta'-\eta'_n)^2},
\end{eqnarray}
where $\eta'_n$ is defined in (\ref{berf_111}), $(a)$ is due to the assumption $\eta'>\eta'_n$ and Chebyshev's inequality and $(b)$ is due to $\mathrm{Var}(\underline{\boldsymbol{\nu}}^{\,T}\Psi\underline{\boldsymbol{\nu}})=2\mathrm{tr}(\Psi^2)$.  
It is shown in Appendix~G that 
\begin{eqnarray}
\label{c_p_n}
\mathrm{tr}(\Psi^2)\leq C'_n/2,
\end{eqnarray}
where $C'_n$ is defined in (\ref{cons_CP}). Then 
\begin{eqnarray}
\label{ahaj}
\Pr\big(|\underline{\boldsymbol{\nu}}^{\,T}\boldsymbol{\Psi}\underline{\boldsymbol{\nu}}-m|\ge m\eta'\,|\,\boldsymbol{H}=H\big)\leq \frac{C'_n}{m^2(\eta'-\eta'_n)^2}.
\end{eqnarray}
The right side of (\ref{ahaj}) does not depend on the choice of $H\in\mathcal{H}$. Then (\ref{yellp_1}), (\ref{doos11_11_D_11}) and (\ref{ahaj}) complete the proof. 
%The proof is slightly different from that of (\ref{nasa_2}). The details appear in Appendix~F. 
\end{proof}
By the previous proposition, we see that if $\lim_{n\to\infty}\frac{1}{n}\lambda_{\max}(\Sigma_n)=0$ and $\eta'>\limsup_{n\to\infty}\eta'_n$, then 
\begin{eqnarray}
\lim_{n\to\infty}\Pr\big(\underline{\boldsymbol{y}}\notin \mathcal{T}_{\eta'}^{(m)}(\Omega_c)\big)=0.
\end{eqnarray} 
Next, we concentrate on the second term on the right side of (\ref{pish_11_11}).  We write
\begin{eqnarray}
\label{comb_1}
&&\Pr\big(\underline{\boldsymbol{x}}_i\in \mathcal{T}_{\varepsilon}^{(n)}(\Sigma_n),\,\underline{\boldsymbol{w}}_i\in \mathcal{T}_{\eta}^{(m+n)}(\Xi_n),\, \underline{\boldsymbol{y}}\in \mathcal{T}_{\eta'}^{(m)}(\Omega_c)\big)\notag\\
&&\hskip0.5cm\stackrel{(a)}{=}\int_{ \mathcal{T}_{\eta}^{(m+n)}(\Xi_n)\bigcap \big(\mathcal{T}_{\varepsilon}^{(n)}(\Sigma_n)\times  \mathcal{T}_{\eta'}^{(m)}(\Omega_c)\big)}p_{\underline{\boldsymbol{x}}_i}(\underline{x})p_{\underline{\boldsymbol{y}}}(\underline{y})\mathrm{d}\underline{x}\,\mathrm{d}\underline{y}\notag\\
%&&\hskip0.5cm\stackrel{(a)}{=}\int_{\mathbb{R}^n\times \mathbb{R}^m}\big(p_{\underline{\boldsymbol{x}}(i)}(\underline{x})\mathds{1}_{\underline{x}\in \mathcal{T}_{\varepsilon}^{(n)}(PI_n)}\big)\big(p_{\underline{\boldsymbol{y}}}(\underline{y})\mathds{1}_{\underline{y}\,\in \mathcal{T}_{\eta'}^{(m)}(\Omega_c)}\big)\mathds{1}_{(\underline{x}^T\,\,\underline{y}^T)^T\in\mathcal{T}_{\eta}^{(m+n)}(\Sigma_n_c)}\mathrm{d}\underline{x}\,\mathrm{d}\underline{y}\notag\\
%&&\hskip0.5cm\stackrel{(b)}{\leq }2^{-h_{\mathsf{G}}(PI_n)+\frac{1}{2\ln 2}n\varepsilon}\times\Big(\sup_{\underline{y}\,\in \mathcal{T}_{\eta'}^{(m)}(\Omega_c)}p_{\underline{\boldsymbol{y}}}(\underline{y})\Big)\times\int_{\mathcal{T}_{\varepsilon}^{(n)}(PI_n)\times  \mathcal{T}_{\eta'}^{(m)}(\Omega_c)}\mathds{1}_{(\underline{x}^T\,\,\underline{y}^T)^T\in\mathcal{T}_{\eta}^{(m+n)}(\Sigma_n_c)}\mathrm{d}\underline{x}\,\mathrm{d}\underline{y}\notag\\
&&\hskip0.5cm\stackrel{(b)}{\leq}2^{-h_{\mathsf{G}}(\Sigma_n)+\frac{1}{2\ln 2}n\varepsilon}\times\sup_{\underline{y}\,\in \mathcal{T}_{\eta'}^{(m)}(\Omega_c)}p_{\underline{\boldsymbol{y}}}(\underline{y})\times\mathrm{Vol}\big( \mathcal{T}_{\eta}^{(m+n)}(\Xi_n)\big),
%&&\hskip0.5cm\stackrel{(b)}{\leq}2^{-h_{\mathsf{G}}(PI_n)+\frac{1}{2\ln 2}n\varepsilon}\times\mathrm{Vol}\big(\mathcal{T}_{\eta}^{(m+n)}(\Xi_n)\big)\sup_{\underline{y}\,\in \mathcal{T}_{\eta'}^{(m)}(\Omega_c)}p_{\underline{\boldsymbol{y}}}(\underline{y})\notag\\
%&&\hskip0.5cm\stackrel{(d)}{\leq}2^{-h_{\mathsf{G}}(PI_n)+\frac{1}{2\ln 2}n\varepsilon}\times 2^{h_{\mathsf{G}}(\Xi_n)+\frac{m+n}{2}\log(1+\eta)}\sup_{\underline{y}\,\in \mathcal{T}_{\eta'}^{(m)}(\Omega_c)}p_{\underline{\boldsymbol{y}}}(\underline{y})\notag\\
%&&\hskip0.5cm\stackrel{(e)}{\leq}2^{\frac{m+n}{2}\log(1+\eta)+\frac{1}{2\ln 2}n\varepsilon}(2\pi e)^{\frac{m}{2}}\sup_{\underline{y}\,\in \mathcal{T}_{\eta'}^{(m)}(\Omega_c)}p_{\underline{\boldsymbol{y}}}(\underline{y}),
%&&\hskip0.5cm\stackrel{(c)}{\leq} 2^{-h_{\mathsf{G}}(PI_n)-h_{\mathsf{G}}(\Omega_c)+\frac{1}{2\ln 2}(n\eta+m\eta')}\,\times\, 2^{h_{\mathsf{G}}(\Sigma_n_c)+\frac{m+n}{2}\log(1+\eta)}\notag\\
%&&\hskip0.5cm\stackrel{(d)}{=} 2^{-\frac{1}{2}\log\det(\Omega_c)+\frac{1}{2\ln 2}(n\eta+m\eta')+ \frac{m+n}{2}\log(1+\eta)},
\end{eqnarray}
where $(a)$ is due to independence of $\underline{\boldsymbol{x}}_i$ and $\underline{\boldsymbol{y}}$ for $i\neq 1$ and $(b)$ is due to $p_{\underline{\boldsymbol{x}}_i}(\underline{x})=p_{\mathsf{G}}(\underline{x};\Sigma_n)\leq 2^{-h_{\mathsf{G}}(\Sigma_n)+\frac{1}{2\ln 2}n\varepsilon}$ for $\underline{x}\in \mathcal{T}_{\varepsilon}^{(n)}(\Sigma_n)$. Note that $p_{\boldsymbol{\underline{y}}}(\underline{y})\neq p_{\mathsf{G}}(\underline{y};\Omega_c)$ and hence, we can not bound $p_{\underline{\boldsymbol{y}}}(\underline{y})$ for $\underline{y}\in \mathcal{T}_{\eta'}^{(m)}(\Omega_c)$ in a similar fashion as we bounded $p_{\underline{\boldsymbol{x}}_i}(\underline{x})$ for $\underline{x}\,\in \mathcal{T}_{\varepsilon}^{(n)}(\Sigma_n)$. To proceed, we need to find upper bounds on the second term and the third term  on the right side of (\ref{comb_1}). Let us begin with the third term. Reference~\cite{Cover-Pombra} provides the upper bound $2^{h_{\mathsf{G}}(\Xi_n)+\frac{m+n}{2\ln2}\eta}$ on the volume of the Gaussian typical set~$\mathcal{T}_\eta^{(m+n)}(\Xi_n)$. This upper bound turns out to be quite loose for our purpose. The parameter $\eta$ will eventually be replaced by $\eta_n$ in (\ref{mir_111}) which scales with $P$. A tighter upper bound is provided by the next lemma.
\newpage
\begin{lem}
\label{seal}
Let $\Sigma$ be an $n\times n$ positive-definite matrix and $\eta>0$. Then
\begin{eqnarray}
\label{sial_1}
\mathrm{Vol}(\mathcal{T}_\eta^{(n)}(\Sigma))\leq 2^{h_{\mathsf{G}}(\Sigma)+\frac{n}{2}\log(1+\eta)}.
\end{eqnarray}   
If $\eta\geq 1$, then this upper bound is tight in the sense that for every sequence of $n\times n$ positive-definite matrices $\Sigma_n$, 
 \begin{eqnarray}
 \label{typ_asy}
\lim_{n\to\infty}\Big(\frac{\mathrm{Vol}(\mathcal{T}_\eta^{(n)}(\Sigma_n))}{2^{h_{\mathsf{G}}(\Sigma_n)}}\Big)^{\frac{2}{n}}=1+\eta.
\end{eqnarray}
\end{lem}
\begin{proof}
The proof is provided in Appendix~H. 
\end{proof}
Applying (\ref{sial_1}) to the $(m+n)\times (m+n)$ matrix $\Xi_n$, 
\begin{eqnarray}
\label{property_1}
\mathrm{Vol}(\mathcal{T}_\eta^{(m+n)}(\Xi_n))\leq 2^{h_{\mathsf{G}}(\Xi_n)+\frac{m+n}{2}\log(1+\eta)}.
\end{eqnarray}
  By (\ref{comb_1}) and (\ref{property_1}), 
 \begin{eqnarray}
 \label{shape}
&&\Pr\big(\underline{\boldsymbol{x}}_i\in \mathcal{T}_{\varepsilon}^{(n)}(\Sigma_n),\,\underline{\boldsymbol{w}}_i\in \mathcal{T}_{\eta}^{(m+n)}(\Xi_n),\, \underline{\boldsymbol{y}}\in \mathcal{T}_{\eta'}^{(m)}(\Omega_c)\big)\notag\\&&\hskip3
cm\leq 2^{h_{\mathsf{G}}(\Xi_n)-h_{\mathsf{G}}(\Sigma_n)+\frac{m+n}{2}\log(1+\eta)+\frac{1}{2\ln 2}n\varepsilon}\sup_{\underline{y}\,\in \mathcal{T}_{\eta'}^{(m)}(\Omega_c)}p_{\underline{\boldsymbol{y}}}(\underline{y}).
\end{eqnarray}
By (\ref{det_11}),
\begin{eqnarray}
h_{\mathsf{G}}(\Xi_n)-h_{\mathsf{G}}(\Sigma_n)=\frac{1}{2}\log\big((2\pi e)^{m+n}\det(\Xi_n)\big)-\frac{1}{2}\log((2\pi e)^n\det(\Sigma_n))=\frac{m}{2}\log(2\pi e).
\end{eqnarray}
 Hence, 
  \begin{eqnarray}
 \label{hi_opk}
&&\Pr\big(\underline{\boldsymbol{x}}_i\in \mathcal{T}_{\varepsilon}^{(n)}(\Sigma_n),\,\underline{\boldsymbol{w}}_i\in \mathcal{T}_{\eta}^{(m+n)}(\Xi_n),\, \underline{\boldsymbol{y}}\in \mathcal{T}_{\eta'}^{(m)}(\Omega_c)\big)\notag\\&&\hskip3cm\leq (2\pi e)^{\frac{m}{2}}2^{\frac{m+n}{2}\log(1+\eta)+\frac{1}{2\ln 2}n\varepsilon}\sup_{\underline{y}\,\in \mathcal{T}_{\eta'}^{(m)}(\Omega_c)}p_{\underline{\boldsymbol{y}}}(\underline{y}).
\end{eqnarray}
   Next, we present an upper bound on the supremum on the right side of (\ref{hi_opk}). We write 
\begin{eqnarray}
\label{denir_11}
p_{\boldsymbol{\underline{y}}}(\underline{y})&=&\int_{\mathcal{H}}p_{\boldsymbol{\underline{y}}}(\underline{y}\,|\,\boldsymbol{H}=H)\mathrm{d}\mathcal{L}_{\boldsymbol{H}}(H)\notag\\
&=&\int_{\mathcal{H}}\frac{1}{(2\pi)^{\frac{m}{2}}\big(\det(I_m+H\Sigma_n H^T)\big)^{\frac{1}{2}}}\exp\Big(-\frac{1}{2}\underline{y}^T(I_m+H\Sigma_n H^T)^{-1}\underline{y}\Big)\mathrm{d}\mathcal{L}_{\boldsymbol{H}}(H)\notag\\
&\leq&(2\pi)^{-\frac{m}{2}}\big(\inf_{H\in \mathcal{H}}\det(\Omega_{H})\big)^{-\frac{1}{2}}\exp\Big(-\frac{1}{2}\inf_{H\in \mathcal{H}}\underline{y}^T\Omega_H^{-1}\underline{y}\Big),
\end{eqnarray}
where $\mathcal{L}_{\boldsymbol{H}}$ is the probability law of $\boldsymbol{H}$ and we have defined
\begin{eqnarray}
\Omega_H:=I_m+H\Sigma_n H^T.
\end{eqnarray} 
 Then 
\begin{eqnarray}
\label{peok_111}
\sup_{\underline{y}\in \mathcal{T}^{(m)}_{\eta'}(\Omega_c)}p_{\boldsymbol{\underline{y}}}(\underline{y})\leq (2\pi)^{-\frac{m}{2}}\big(\inf_{H\in \mathcal{H}}\det(\Omega_{H})\big)^{-\frac{1}{2}}\exp\Big(-\frac{1}{2}\inf_{H\in \mathcal{H}}\inf_{\underline{y}\in \mathcal{T}^{(m)}_{\eta'}(\Omega_c)}\underline{y}^T\Omega_H^{-1}\underline{y}\Big).
\end{eqnarray}
The next two lemmas provide lower bounds on the infimums on the right side of~(\ref{peok_111}): 
\begin{lem}
\label{lem_55}
Recall $\Omega_c$ in (\ref{omega_c}) and $\phi_{1,n}$ in (\ref{joda_22}). Then
\begin{eqnarray}
\label{zcb_11}
\inf_{H\in \mathcal{H}}\det(\Omega_{H})\geq (1-\phi_{1,n} )^m\det(\Omega_c).
\end{eqnarray}
%\begin{eqnarray}
%\label{zcb_11111}
%\inf_{H\in \mathcal{H}}\det(\Omega_{H})\geq 2^{-2m\phi_{1,n}}\det(\Omega_c).
%\end{eqnarray}
% where it is assumed  that $\phi_4< 1$.
\end{lem}
\begin{proof}
The proof relies on Weyl's inequality on perturbation of eigenvalues of a symmetric matrix. See Appendix~I for the details. 
\end{proof}
 \begin{lem}
 \label{lem_66}
 Recall $\phi_{3,n}$ in (\ref{nmkio}). Then
 \begin{eqnarray}
 \label{zcb_22}
\inf_{H\in \mathcal{H}}\inf_{\underline{y}\in \mathcal{T}^{(m)}_{\eta'}(\Omega_c)}\underline{y}^T\Omega_H^{-1}\underline{y}\geq m(1-\eta')^+\phi_{3,n}.
\end{eqnarray}
 \end{lem}
 \begin{proof}
 The inner optimization is a quadratically constrained quadratic convex program. See Appendix~J for the details.
 \end{proof}
 By (\ref{peok_111}), (\ref{zcb_11}) and (\ref{zcb_22}),
 \begin{eqnarray}
 \label{vio_111}
\sup_{\underline{y}\in \mathcal{T}^{(m)}_{\eta'}(\Omega_c)}p_{\boldsymbol{\underline{y}}}(\underline{y})\leq (2\pi)^{-\frac{m}{2}}2^{-\frac{1}{2}\log\det(\Omega_c)}2^{-\frac{m}{2}\log(1-\phi_{1,n})}\exp\big(-\frac{m}{2}(1-\eta')^+\phi_{3,n}\big),%%%%%2^{-\frac{1}{2}\log\det(\Omega_c)-\frac{m}{2}\log(1-\phi_n)-\frac{m}{2}(1-\eta')\psi_n\log e}
\end{eqnarray}
Putting (\ref{hi_opk}) and (\ref{vio_111}) together, 
\begin{eqnarray}
\label{khoonam_1}
\sum_{i=2}^{2^{nR}}\Pr\big(\underline{\boldsymbol{x}}_i\in \mathcal{T}_{\varepsilon}^{(n)}(\Sigma_n),\,\underline{\boldsymbol{w}}_i\in \mathcal{T}_{\eta}^{(m+n)}(\Xi_n),\, \underline{\boldsymbol{y}}\in \mathcal{T}_{\eta'}^{(m)}(\Omega_c)\big)\leq 2^{n(R-\frac{1}{2n}\log\det(\Omega_c)+\frac{m+n}{2n}\log(1+\eta)+\tilde{\delta}_{n})},
\end{eqnarray}%(\ref{pish_11_11}), (\ref{peter_111}), (\ref{comb_1}) and (\ref{vio_111}) together, the probability of type~II error is bounded by 
%\begin{eqnarray}
%\label{oto_1}
%\Pr(\mathcal{E}^{(II)})\leq C'/(m(\eta'-\eta'_n)^2)+2^{n(R-\frac{1}{2n}\log\det(\Omega_c)+\delta_n)},
%\end{eqnarray} 
where  $\tilde{\delta}_{n}$ is given by
\begin{eqnarray}
\label{jo_1}
\tilde{\delta}_{n}=-\frac{m}{2n}\log(1-\phi_{1,n})+\frac{m}{2n\ln 2}(1-(1-\eta')^+\phi_{3,n})+\frac{\varepsilon}{2\ln 2},
\end{eqnarray}  
and we have the conditions $\eta'>\eta'_n$ and $\eta>\eta_n$. Letting $\eta$, $\eta'$ and $\varepsilon$ approach $\eta_n$, $\eta'_n$ and $0$, from above, respectively, we see that the right side of (\ref{khoonam_1}) vanishes as~$n$ grows large if 
\begin{eqnarray}
\label{nagoo_0}
R<\lim_{n\to\infty}\Big(\frac{1}{2n}\log\det(\Omega_c)-\log(1+\eta_n)-\delta_n\Big),
\end{eqnarray}
where $\delta_n$ is given in (\ref{kool}).

\section{Summary and concluding remarks}
We studied a stochastic and time varying Gaussian ISI channel where the $i^{th}$ tap during time slot~$t$ is a random variable whose probability law is supported over an interval of radius $r_i$ centred at $c_i$. The joint distribution as well as realizations of the array of channel taps are unknown to both ends of communication. A lower bound was derived on the channel capacity by carefully working out the details of error analysis for a decoder which functions based on Gaussian joint-typicality decoding tuned to the matrix~$H_c$. The proposed lower bound saturates at a positive value as $P$ increases beyond the saturation power $P_{sat}$ given in (\ref{sat_p}). A partial converse result was presented that shows for a sequence of codebooks with vanishingly small probability of error, if the size of each symbol in every codeword is bounded away from zero by an amount that is proportional to $\sqrt{P}$, then the rate of the codebooks does not scale with $P$. This converse result holds in the worst-case scenario where the channel taps are independent and uniformly distributed.  

In view of this converse result, a question rises naturally whether one can avoid saturation of rates by inserting zero symbols in the codewords. In the context of fast fading channels where the transmitter and the receiver have not access to the channel state information, reference \cite{abou-faycal} proves that the capacity-achieving input distribution has a mass point at zero. Even though the channel model in~\cite{abou-faycal} and the channel model adopted here are quite different, the idea of randomly generating codewords according to a distribution that has a mass point at zero is worth investigating towards the possibility of achieving rates that scale (do not saturate) with the maximum average input power $P$.

An observation made in the paper which may find other applications is the result in Lemma~\ref{seal} regarding the volume of Gaussian typical sets. It shows that 
\begin{enumerate}
  \item For every $\eta>0$, $\mathrm{Vol}(\mathcal{T}_\eta^{(n)}(\Sigma))\leq 2^{h_{\mathsf{G}}(\Sigma)+\frac{n}{2}\log(1+\eta)}$.
  \item For every sequence of positive definite matrices $\Sigma_n$ and every $\eta\geq1$, the upper bound in above is tight in the sense that 
\begin{eqnarray}
\lim_{n\to\infty}\Big(\frac{\mathrm{Vol}(\mathcal{T}^{(n)}_\eta(\Sigma_n))}{2^{h_{\mathsf{G}}(\Sigma_n)}}\Big)^{\frac{2}{n}}=1+\eta.
\end{eqnarray}
\end{enumerate}

We close this section by mentioning that the decoding rule adopted in this paper is reminiscent of the notion of ``mismatched decoding'' which is extensively studied for memoryless channels~\cite{scarlet}. Characterizing fundamental limits of mismatched decoding for time-varying channels with memory is another direction for investigation.   

\section*{Appendix~A; Proof of Theorem~2}
Let $\mathcal{C}_n=\{\underline{x}_{n,1},\cdots, \underline{x}_{n,2^{nR}}\}$ be  a sequence of codebooks with rate $R$ and vanishingly small probability of error. Throughout this appendix, we will drop the codebook index $n$ and denote $\underline{x}_{n,i}$ by $\underline{x}_i$ for notational simplicity. The $j^{th}$ symbol of $\underline{x}_{i}$ is denoted by $x_{i,j}$. The codebook satisfies the average transmission power constraint in (\ref{ghor_1}), i.e., 
\begin{eqnarray}
\frac{1}{2^{nR}}\sum_{i=1}^{2^{nR}}\|\underline{x}_i\|_2^2\leq nP. 
\end{eqnarray}
\label{polt}
The transmitted vector $\boldsymbol{\underline{x}}$ in (\ref{jat_123}) is uniformly distributed over $\mathcal{C}_n$.  A standard application of Fano's inequality and the data processing inequality gives\footnote{See \cite{Cover-Thomas}.} 
 \begin{eqnarray}
 \label{fanot}
R\leq \liminf_{n\to\infty}\frac{1}{n}I(\underline{\boldsymbol{x}};\underline{\boldsymbol{y}}).
\end{eqnarray}
 %   Then
%\begin{eqnarray}
%\label{mos_00}
%nR&=&H(\boldsymbol{W})\notag\\
%&=&H(\boldsymbol{W}|\underline{\boldsymbol{y}})+I(\boldsymbol{W};\underline{\boldsymbol{y}})\notag\\
%&\le&1+nRp_{e}(\mathcal{C}_n)+I(\underline{\boldsymbol{x}};\underline{\boldsymbol{y}}),
%&=&1+np_{e}(\mathcal{C})R+h(\underline{\boldsymbol{y}})-h(\underline{\boldsymbol{y}}|\underline{\boldsymbol{x}}),
%\end{eqnarray}
%where the last step is due to the Fano's inequality and the data processing inequality for mutual information. Let us write
We write $I(\underline{\boldsymbol{x}};\underline{\boldsymbol{y}})= h(\underline{\boldsymbol{y}})- h(\underline{\boldsymbol{y}}|\underline{\boldsymbol{x}})$.  We find an upper bound and a lower bound on the terms $h(\underline{\boldsymbol{y}})$ and $h(\underline{\boldsymbol{y}}|\underline{\boldsymbol{x}})$, respectively. 
\subsection{Upper bound on $h(\underline{\boldsymbol{y}})$}
Applying the maximum entropy lemma~\cite{Cover-Thomas}, 
\begin{eqnarray}
\label{lll_1}
h(\underline{\boldsymbol{y}})\le\frac{1}{2}\log\big((2\pi e)^m\mathrm{det}(\mathrm{Cov}(\underline{\boldsymbol{y}}))\big).%-\frac{1}{2^{nR}}\sum_{i=1}^{2^{nR}}h(\underline{\boldsymbol{y}}|\underline{\boldsymbol{x}}=\underline{x}(i)).
\end{eqnarray}
%By Jensen's inequality and  concavity of $\log\det(\cdot)$, 
%\begin{eqnarray}
%\label{mellat1}
%\langle h(\underline{\boldsymbol{y}})\rangle_{\mathcal{C}}\le\frac{1}{2}\log\big((2\pi e)^m\mathrm{det}(\langle \mathrm{Cov}(\underline{\boldsymbol{y}})\rangle_{\mathcal{C}})\big)
%\end{eqnarray}  
Define 
\begin{eqnarray}
\label{ghol_00}
\underline{\mu}:=\mathbb{E}[\underline{\boldsymbol{x}}]=\frac{1}{2^{nR}}\sum_{i=1}^{2^{nR}}\underline{x}_i,\,\,\,\,\,Q:=\mathbb{E}[\underline{\boldsymbol{x}}\,\underline{\boldsymbol{x}}^T]=\frac{1}{2^{nR}}\sum_{i=1}^{2^{nR}}\underline{x}_i\underline{x}^T_i.
\end{eqnarray}
A simple computation shows that 
\begin{eqnarray}
\label{fol_1}
\mathrm{Cov}(\underline{\boldsymbol{y}})=I_m+\mathbb{E}[\boldsymbol{H}Q\boldsymbol{H}^T]-\mathbb{E}[\boldsymbol{H}]\underline{\mu}\,\underline{\mu}^T\mathbb{E}[\boldsymbol{H}^T]
\end{eqnarray}
By (\ref{lll_1}) and (\ref{fol_1}) and using the fact that $\log\det(\cdot)$ is nondecreasing along the cone of positive-semidefinite matrices\footnote{We have $\mathrm{Cov}(\underline{\boldsymbol{y}})\preceq I_m+\mathbb{E}[\boldsymbol{H}Q\boldsymbol{H}^T]$ by (\ref{fol_1}).},  
\begin{eqnarray}
\label{gub}
h(\underline{\boldsymbol{y}})\le\frac{1}{2}\log\big((2\pi e)^m\mathrm{det}(I_m+\mathbb{E}[\boldsymbol{H}Q\boldsymbol{H}^T])\big).
\end{eqnarray}
Applying the arithmetic-geometric inequality $\mathrm{det}(M)\leq(\frac{\mathrm{tr}(M)}{m})^m$ for an $m\times m$ positive semidefinite matrix~$M$,
\begin{eqnarray}
\label{expo_11}
h(\underline{\boldsymbol{y}})\leq \frac{m}{2}\log(2\pi e)+\frac{1}{2}m\log\Big(1+\frac{1}{m}\sum_{t=1}^m\mathbb{E}\big[\big[\boldsymbol{H}Q\boldsymbol{H}^T\big]_{t,t}\big]\Big).
\end{eqnarray} 
Next, we find an upper bound on the diagonal entries of the $m\times m$ matrix $\mathbb{E}[\boldsymbol{H}Q\boldsymbol{H}^T]$. For $1\leq i\leq m$, 
\begin{eqnarray}
\label{expo_11_11}
[\boldsymbol{H}Q\boldsymbol{H}^T]_{t,t}&=&\sum_{u,v=1}^n\boldsymbol{H}_{t,u}Q_{u,v}\boldsymbol{H}_{t,v}\notag\\
&\stackrel{(a)}{=}&\Big|\sum_{u,v=1}^n\boldsymbol{H}_{t,u}Q_{u,v}\boldsymbol{H}_{t,v}\Big|\notag\\
&\stackrel{(b)}{\leq}&\sum_{u,v=1}^n|\boldsymbol{H}_{t,u}||\boldsymbol{H}_{t,v}||Q_{u,v}|\notag\\
&\stackrel{(c)}{\leq}&\sum_{u,v=1}^n|\boldsymbol{H}_{t,u}||\boldsymbol{H}_{t,v}|\sqrt{Q_{u,u}Q_{v,v}}\notag\\
&\stackrel{}{=}&\sum_{\substack{u,v=1\\
0\leq t-u,t-v\leq k}}^n|\boldsymbol{h}_{t,t-u}||\boldsymbol{h}_{t,t-v}|\sqrt{Q_{u,u}Q_{v,v}}\notag\\
&\stackrel{}{=}&\Big(\sum_{\substack{u=1\\
0\leq t-u\leq k}}^n|\boldsymbol{h}_{t,t-u}|\sqrt{Q_{u,u}}\Big)^2\notag\\
&\stackrel{(d)}{\leq}&\Big(\sum_{\substack{u=1\\
0\leq t-u\leq k}}^n|\boldsymbol{h}_{t,t-u}|^2\Big)\Big(\sum_{\substack{u=1\\
0\leq t-u\leq k}}^nQ_{u,u}\Big),
\end{eqnarray}
where $(a)$ is due $[\boldsymbol{H}Q\boldsymbol{H}^T]_{t,t}$ being nonnegative, $(b)$ is due to the triangle inequality, $(c)$ is due to the fact that $Q$ is a positive semidefinite matrix and hence, each of its central $2\times 2$ minors are nonnegative, i.e., $Q_{u,u}Q_{v,v}-Q^2_{u,v}\geq 0$ and $(d)$ is due to Cauchy-Schwarz inequality. %$(\sum_{u=0}^ka_u)^2\leq (k+1)\sum_{u=0}^ka_u^2$. 
Taking expectations of both sides of (\ref{expo_11_11}), 
\begin{eqnarray}
\label{expo_9080}
\mathbb{E}[[\boldsymbol{H}Q\boldsymbol{H}^T]_{t,t}]&\leq&\Big(\sum_{\substack{u=1\\
0\leq t-u\leq k}}^n\mathbb{E}[|\boldsymbol{h}_{t,t-u}|^2]\Big)\Big(\sum_{\substack{u=1\\
0\leq t-u\leq k}}^nQ_{u,u}\Big)\notag\\
&=&\Big(\sum_{\substack{u=1\\
0\leq t-u\leq k}}^n(c_{t-u}^2+r_{t-u}^2/3)\Big)\Big(\sum_{\substack{u=1\\
0\leq t-u\leq k}}^nQ_{u,u}\Big)\notag\\
&\leq&\big(\|\underline{c}\|_2^2+\|\underline{r}\|_2^2/3\big)\sum_{\substack{u=1\\
0\leq t-u\leq k}}^nQ_{u,u}
\end{eqnarray}
where the penultimate step is due to $\boldsymbol{h}_{i,i-u}$ being uniformly distributed on an interval of radius $r_{i-u}$ centred at~$c_{i-u}$ and the last step is due to $\sum_{\substack{u=1\\
0\leq t-u\leq k}}^n(c_{t-u}^2+r_{t-u}^2/3)\leq \|\underline{c}\|_2^2+\|\underline{r}\|_2^2/3$. By (\ref{expo_11}) and (\ref{expo_9080}), 
\begin{eqnarray}
h(\underline{\boldsymbol{y}})\leq\frac{m}{2}\log(2\pi e)+\frac{1}{2}m\log\bigg(1+\frac{1}{m}\big(\|\underline{c}\|_2^2+\|\underline{r}\|_2^2/3\big)\sum_{t=1}^m\sum_{\substack{u=1\\
0\leq t-u\leq k}}^nQ_{u,u}\bigg).
\end{eqnarray}
Changing the order of summations show
\begin{eqnarray}
\sum_{t=1}^m\sum_{\substack{u=1\\
0\leq t-u\leq k}}^nQ_{u,u}&=&\sum_{u=1}^n\sum_{\substack{t=1\\
0\leq t-u\leq k}}^mQ_{u,u}\notag\\
&=&\sum_{u=1}^nQ_{u,u}\sum_{\substack{t=1\\
0\leq t-u\leq k}}^m1\notag\\
&\leq&(k+1)nP,
\end{eqnarray} 
where the last step is due to $\sum_{\substack{t=1\\
0\leq t-u\leq k}}^m1\leq k+1$ and $\mathrm{tr}(Q)\leq nP$. We have shown that 
\begin{eqnarray}
\label{beheshtp_1}
h(\underline{\boldsymbol{y}})\leq\frac{m}{2}\log(2\pi e)+\frac{1}{2}m\log\Big(1+(k+1)\big(\|\underline{c}\|_2^2+\|\underline{r}\|_2^2/3\big)\frac{nP}{m}\Big).
\end{eqnarray}

\subsection{Lower bound on $h(\underline{\boldsymbol{y}}|\underline{\boldsymbol{x}}) $}

We have
\begin{eqnarray}
\label{goosh1}
h(\underline{\boldsymbol{y}}|\underline{\boldsymbol{x}}) =\frac{1}{2^{nR}}\sum_{i=1}^{2^{nR}}h(\underline{\boldsymbol{y}}|\underline{\boldsymbol{x}}=\underline{x}_i). 
\end{eqnarray}  
To develop a lower bound on $h(\underline{\boldsymbol{y}}|\underline{\boldsymbol{x}}=\underline{x}_i)$, we write
\begin{eqnarray}
\label{cvb_111}
h(\underline{\boldsymbol{y}}|\underline{\boldsymbol{x}}=\underline{x}_i)&=&h(\boldsymbol{H}\underline{\boldsymbol{x}}+\underline{\boldsymbol{z}}|\underline{\boldsymbol{x}}=\underline{x}_i)\notag\\
&\stackrel{(a)}{=}&h(\boldsymbol{H}\underline{x}_i+\underline{\boldsymbol{z}})\notag\\
&\stackrel{(b)}{=}&h(\boldsymbol{H}\underline{x}_i-H_c\,\underline{x}_i+\underline{\boldsymbol{z}})\notag\\
&\stackrel{(c)}{=}&h(\boldsymbol{E}\,\underline{x}_i+\underline{\boldsymbol{z}})\notag\\
&\stackrel{(d)}{=}&\sum_{t=1}^m h([\boldsymbol{E}\,\underline{x}_i]_t+\boldsymbol{z}_t),
\end{eqnarray}
where $(a)$ is due to independence of $\underline{\boldsymbol{x}}$ and the pair $(\boldsymbol{H},\underline{\boldsymbol{z}})$, $(b)$ is due to $H_c\,\underline{x}$ being a deterministic vector, the matrix $\boldsymbol{E}=\boldsymbol{H}-H_c$ in $(c)$ was originally defined in (\ref{err_111}) and $(d)$ follows the independence of the entries of the vector $\boldsymbol{E}\,\underline{x}+\underline{\boldsymbol{z}}$. Moreover, 
\begin{eqnarray}
\label{cvb_222}
h([\boldsymbol{E}\,\underline{x}_i]_t+\boldsymbol{z}_t)&\stackrel{}{=}&h\Big(\sum_{j=1}^n[\boldsymbol{E}]_{t,j}\,[\underline{x}_{i}]_{j}+\boldsymbol{z}_t\Big)\notag\\&\stackrel{(a)}{=}&h\Big(\sum_{\substack{j=1\\ 0\le t-j\le k}}^{n}[\boldsymbol{E}]_{t,j}\,x_{i,j}+\boldsymbol{z}_t\Big)\notag\\
&\stackrel{(b)}{\geq}&\frac{1}{2}\log\Big(\sum_{\substack{j=1\\ 0\le t-j\le k}}^{n}2^{2h([\boldsymbol{E}]_{t,j}x_{i,j})}+2^{2h(\boldsymbol{z}_t)}\Big)\notag\\
&\stackrel{(c)}{=}&\frac{1}{2}\log\Big(\sum_{\substack{j=1\\ 0\le t-j\le k}}^{n}2^{2\log(2r_{t-j}|x_{i,j}|)}+2\pi e\Big)\notag\\
&=&\frac{1}{2}\log\Big(\sum_{\substack{j=1\\ 0\le t-j\le k}}^{n}4r^2_{t-j}x^2_{i,j}+2\pi e\Big)\notag\\
&\stackrel{(d)}{=}&\frac{1}{2}\log\Big(\sum_{j=0}^{k}4r^2_{j}x^2_{i,t-j}+2\pi e\Big),
\end{eqnarray}
where in $(a)$ we have denoted the $j^{th}$ entry of the codeword $\underline{x}_i$ by $x_{i,j}$, $(b)$ is due to the entropy power inequality~\cite{Cover-Thomas} and the independence of the random variables $[\boldsymbol{E}]_{t,j}$ for $j=t-k,\cdots, t$ and $\boldsymbol{z}_t$, $(c)$ is due to $h(\boldsymbol{z}_t)=\frac{1}{2}\log(2\pi e)$ and the fact that $[\boldsymbol{E}]_{t,j}x_{i,j}=(\boldsymbol{h}_{t,t-j}-c_{t-j})x_{i,j}$ is uniformly distributed over the interval $\big[-r_{t-j}|x_{i,j}|,r_{t-j}|x_{i,j}|\,\big]$ for every $t,j$ with $0\leq t-j\leq k$ and hence, $h([\boldsymbol{E}]_{t,j}x_{i,j})=\log(2r_{t-j}|x_{i,j}|)$ and $(d)$ follows after changing the index $j$ to $t-j$ where it is understood that $x_{i,t-j}=0$ for $t-j<1$ or $t-j>n$. By~(\ref{goosh1}), (\ref{cvb_111}) and (\ref{cvb_222}),

%By entropy power inequality~\cite{Cover-Thomas} and independence of $\boldsymbol{E}$ and $\boldsymbol{z}$, 
%\begin{eqnarray}
%\label{cvb_222}
%h([\boldsymbol{E}\,\underline{x}_t]_i+\boldsymbol{z}_i)&\geq&\frac{1}{2}\log\Big(\sum_{u=i-k}^{i}2^{2h([\boldsymbol{E}]_{i,u}x_{t,u})}+2^{2h(\boldsymbol{z}_i)}\Big)\notag\\
%&\stackrel{}{=}&\frac{1}{2}\log\Big(\sum_{u=i-k}^{i}2^{2\log(2r_{i-u}|x_{t,u}|)}+2\pi e\Big)\notag\\
%&=&\frac{1}{2}\log\Big(\sum_{u=i-k}^{i}4r^2_{i-u}x^2_{t,u}+2\pi e\Big),
%&\stackrel{(b)}{\geq}&\frac{1}{2}\max\Big\{\log\sum_{j=0}^{k}4r^2_{j}x^2_{i-j},\log(2\pi e)\Big\}\notag\\
%&\stackrel{(c)}{=}&\frac{1}{2}\log(2\pi e)+\frac{1}{2}\Big(\log\sum_{j=0}^{k}\frac{2r_j^2x^2_{i-j}}{\pi e}\Big)^+
%\end{eqnarray}
%where the penultimate step is due to $h(\boldsymbol{z}_i)=\frac{1}{2}\log(2\pi e)$ and the fact that $[\boldsymbol{E}]_{i,u}x_{t,u}=(\boldsymbol{h}_{i,i-u}-c_{i-u})x_{t,u}$ is uniformly distributed over the interval $\big[-r_{i-u}|x_{t,u}|,r_{i-u}|x_{t,u}|\,\big]$ for every $i,u$ with $0\leq i-u\leq k$. By~(\ref{cvb_111}) and (\ref{cvb_222}), 
\begin{eqnarray}
\label{bert0}
h(\underline{\boldsymbol{y}}|\underline{\boldsymbol{x}})\geq\frac{1}{2}m\log(2\pi e)+\frac{1}{2^{nR}}\sum_{i=1}^{2^{nR}}\sum_{t=1}^m \frac{1}{2}\log\Big(1+\frac{2}{\pi e}\sum_{j=0}^{k}r^2_{j}x^2_{i,t-j}\Big).
\end{eqnarray}
%&=&\frac{1}{2}\sum_{i=k+1}^n\mathbb{E}\Big[ \log\Big(\sum_{k=0}^{k}4r^2_{k}\boldsymbol{x}^2_{i-k}+2\pi e\Big)\Big]\notag\\
%&=&\frac{1}{2}\sum_{i=k+1}^n\mathbb{E}\Big[ \log\Big(\sum_{k=0}^{k}\Big(r^2_{k}\boldsymbol{x}^2_{i-k}+\frac{2\pi e}{k+1}\Big)\Big)\Big]\notag\\
%&\stackrel{(a)}{\ge}&\frac{1}{2}\sum_{i=k+1}^n\mathbb{E}\Big[ \log\Big((k+1)\Big(\prod_{k=0}^{k}\Big(r^2_{k}\boldsymbol{x}^2_{i-k}+\frac{2\pi e}{k+1}\Big)\Big)^{\frac{1}{k+1}}\Big)\Big]\notag\\
%&=&\frac{1}{2}(n-k)\log(k+1)+\frac{1}{2(k+1)}\sum_{i=k+1}^n\sum_{k=0}^{k}\mathbb{E}\Big[\log\Big(r^2_{k}\boldsymbol{x}^2_{i-k}+\frac{2\pi e}{k+1}\Big)\Big]\notag\\
%&\stackrel{(b)}{\geq}&\frac{1}{2}(n-k)\log(k+1)+\frac{n-k}{2(k+1)}\sum_{k=0}^{k}\mathbb{E}\Big[\log\Big(r^2_{k}\boldsymbol{x}^2_{k+1-k}+\frac{2\pi e}{k+1}\Big)\Big]\notag\\
%&\stackrel{(b)}{\geq}&\frac{1}{2}(n-k)\log(k+1)+\frac{n-k}{2(k+1)}\sum_{k=0}^{k}\mathbb{E}\Big[\frac{1}{2}+\frac{1}{2}\log(r^2_{k}\boldsymbol{x}^2_{k+1-k})+\frac{1}{2}\log\frac{2\pi e}{k+1}\Big)\Big]\notag\\
%\end{eqnarray}
Putting together the upper bound on $h(\underline{\boldsymbol{y}})$ in (\ref{beheshtp_1}) and the lower bound on $h(\underline{\boldsymbol{y}}|\underline{\boldsymbol{x}})$ in (\ref{bert0}),  
%\begin{eqnarray}
%\label{sloth}
%I(\underline{\boldsymbol{x}};\underline{\boldsymbol{y}})&\leq& \frac{1}{2}m\log\Big(1+\frac{1}{m}\sum_{i=1}^m\sum_{\substack{u=1\\ 0\leq i-u\leq k}}^n (2c^2_{i-u}+r_{i-u}^2/3)Q_{u,u}\Big)\notag\\&&\hskip1.5cm-\frac{1}{2^{nR}}\sum_{t=1}^{2^{nR}}\sum_{i=1}^m \frac{1}{2}\log\Big(1+\frac{2}{\pi e}\sum_{\substack{u=1\\ 0\le i-u\le k}}^{n}r^2_{i-u}x^2_{t,u}\Big).\end{eqnarray}
%Changing , we can write 
\begin{eqnarray}
\label{expit_1}
I(\underline{\boldsymbol{x}};\underline{\boldsymbol{y}})&\leq&\frac{1}{2}m\log\Big(1+(k+1)\big(\|\underline{c}\|_2^2+\|\underline{r}\|_2^2/3\big)\frac{nP}{m}\Big).\notag\\&&\hskip1.5cm-\frac{1}{2^{nR}}\sum_{i=1}^{2^{nR}}\sum_{t=1}^m \frac{1}{2}\log\Big(1+\frac{2}{\pi e}\sum_{j=0}^{k}r^2_{j}x^2_{i,t-j}\Big).
\end{eqnarray}
 Dividing both sides by $n$ and letting $n$ grow to infinity, we arrive at the bound in (\ref{ub_ub}).

\section*{Appendix~B; Proof of Proposition~1}
Fix $\varepsilon>0$. We maximize $\frac{1}{2n}\sum_{i=1}^n\log(1+\lambda_{n,i}d_{n,i})-\log(1+(k+1)\|\underline{r}\|_2^2\frac{\sum_{i=1}^nd_{n,i}}{m+n})$ subject to the conditions $d_{n,i}\geq \varepsilon$ and $\sum_{i=1}^n d_{n,i}\leq nP$. The Lagrangian is given by
\begin{eqnarray}
L&=&\frac{1}{2n}\sum_{i=1}^n\log(1+\lambda_{n,i}d_{n,i})-\log\Big(1+(k+1)\|\underline{r}\|_2^2\frac{\sum_{i=1}^nd_{n,i}}{m+n}\Big)\notag\\
&&\hskip5cm+\sum_{i=1}^n\mu_i(d_i-\varepsilon)+\mu\big(nP-\sum_{i=1}^n d_{n,i}\big),
\end{eqnarray}
where the Lagrange multipliers $\mu_1,\cdots, \mu_n,\mu\geq 0$ and the complementary slackness conditions
\begin{eqnarray}
\forall\, i, \,\,\,\mu_i(d_{n,i}-\varepsilon)=0,\,\,\,\mu\big(nP-\sum_{i=1}^n d_{n,i}\big)=0
\end{eqnarray}  
hold.  The first order necessary conditions $\frac{\partial L}{\partial d_{n,i}}=0$ give 
\begin{eqnarray}
\label{piy}
\frac{\frac{1}{2\ln 2}\frac{1}{n}}{\frac{1}{\lambda_{n,i}}+d_{n,i}}-\frac{\frac{1}{\ln 2}\frac{k+1}{m+n}\|\underline{r}\|_2^2}{1+(k+1)\|\underline{r}\|_2^2\frac{\sum_{i=1}^nd_{n,i}}{m+n}}+\mu_i-\mu=0.
\end{eqnarray}
For all $i$ such that $d_{n,i}>\varepsilon$, we must have $\mu_i=0$ and then (\ref{piy}) shows that $\frac{1}{\lambda_{n,i}}+d_{n,i}$ does not depend on the index $i$, i.e., it is a constant $\Theta$ and we get
\begin{eqnarray}
\label{iran}
d_{n,i}=\max\Big(\Theta-\frac{1}{\lambda_{n,i}},\varepsilon\Big)
\end{eqnarray} 
Two situations can happen as we describe next: 
\begin{enumerate}
  \item Assume $\sum_{i=1}^n d_{n,i}=nP$. Then $\Theta$ solves the equation $\sum_{i=1}^n\max(\Theta-\frac{1}{\lambda_{n,i}},\varepsilon)=nP$. Dividing both sides by $n$ and letting $n$ grow to infinity, Szeg\"o's Theorem gives
  \begin{eqnarray}
  \label{ziba_1}
  \frac{1}{2\pi}\int_0^{2\pi}\max\big(\Theta-|f(\omega)|^{-2},\varepsilon\big)\mathrm{d}\omega=P.
\end{eqnarray}
  \item Assume $\sum_{i=1}^n d_{n,i}<nP$. Then $\mu=0$ and (\ref{piy}) gives 
  \begin{eqnarray}
  \label{piy_2}
\frac{\frac{1}{2\ln 2}\frac{1}{n}}{\frac{1}{\lambda_{n,i}}+d_{n,i}}-\frac{\frac{1}{\ln 2}\frac{k+1}{m+n}\|\underline{r}\|_2^2}{1+(k+1)\|\underline{r}\|_2^2\frac{\sum_{i=1}^nd_{n,i}}{m+n}}+\mu_i=0.
\end{eqnarray}
If there is at least one $i$ such that $d_{n,i}>\varepsilon$, then $\mu_i=0$ and (\ref{piy_2}) further simplifies to   
\begin{eqnarray}
\sum_{i=1}^n\max\Big(\Theta-\frac{1}{\lambda_{n,i}},\varepsilon\Big)=2n\Theta-\frac{m+n}{(k+1)}\frac{1}{\|\underline{r}\|_2^2}.
\end{eqnarray}
Dividing both sides by $n$, letting $n$ grow to infinity and invoking Szeg\"o's Theorem, 
\begin{eqnarray}
\label{ziba_2}
  \frac{1}{2\pi}\int_0^{2\pi}\max\big(\Theta-|f(\omega)|^{-2},\varepsilon\big)\mathrm{d}\omega=2\Theta-\frac{2}{(k+1)}\frac{1}{\|\underline{r}\|_2^2}.
\end{eqnarray}
\end{enumerate}
We also need to determine $\lim_{n\to\infty}\delta_n$. This requires computing the limiting values for $d_{n,1}=\max(\Theta-\frac{1}{\lambda_{n,1}},\varepsilon)$ and $d_{n,n}=\max(\Theta-\frac{1}{\lambda_{n,n}},\varepsilon)$. We have $\lambda_{n,1}=\lambda_{\min}(H_c^TH_c)$ and $\lambda_{n,n}=\lambda_{\max}(H_c^TH_c)$. The spectral function for the Toeplitz matrix $H_c^TH_c$ is $|f(\omega)|^2$.  Then Corollary~4.2 on page~58 in \cite{Gray} gives 
\begin{eqnarray}
\lim_{n\to\infty}\lambda_{n,1}=\alpha^2,\,\,\,\,\,\lim_{n\to\infty}\lambda_{n,n}=\beta^2.
\end{eqnarray}
It is also clear that the condition in (\ref{vgt_22}) holds. In fact,
\begin{eqnarray}
\lim_{n\to\infty}\frac{1}{n}\lambda_{\max}(\Sigma_n)=\lim_{n\to\infty}\frac{d_{n,n}}{n}=\lim_{n\to\infty}\frac{1}{n}\max(\Theta-\beta^{-2},\varepsilon)=0.
\end{eqnarray}
  Finally, letting $\varepsilon$ approach zero from above, the proof of Proposition~1 is complete.

\section*{Appendix~C; Proof of (\ref{c_n})}
For a square block matrix $M=\begin{bmatrix}
   A  &  B  \\
    C  & D 
\end{bmatrix}$, we have $\mathrm{tr}(M^2)=\mathrm{tr}(A^2+BC+CB+D^2)$. Applying this to $\Phi$ in (\ref{goli_111}), we get 
\begin{eqnarray}
\mathrm{tr}(\Phi^2)&=&\mathrm{tr}((I_n+\Sigma_n^{1/2}E^TE\Sigma_n^{1/2})^2+\Sigma_n^{1/2}E^TE\Sigma_n^{1/2}+E\Sigma_nE^T+I_m)\notag\\
&=&\mathrm{tr}(I_n+2\Sigma_n^{1/2}E^TE\Sigma_n^{1/2}+(\Sigma_n^{1/2}E^TE\Sigma_n^{1/2})^2+\Sigma_n^{1/2}E^TE\Sigma_n^{1/2}+E\Sigma_nE^T+I_m).
\end{eqnarray} 
Using the identity $\mathrm{tr}(M_1M_2)=\mathrm{tr}(M_2M_1)$,  
%\begin{eqnarray}
%\mathrm{tr}((\Sigma_n^{1/2}E^TE\Sigma_n^{1/2})^2)=\mathrm{tr}((E^TE\Sigma_n)^2),
%\end{eqnarray}
\begin{eqnarray}
\mathrm{tr}(\Sigma_n^{1/2}E^TE\Sigma_n^{1/2})= \mathrm{tr}(E^TE\Sigma_n),\,\,\,\,\,\mathrm{tr}(E\Sigma_nE^T)=  \mathrm{tr}(E^TE\Sigma_n).
\end{eqnarray}
%and 
%\begin{eqnarray}
%\mathrm{tr}(\boldsymbol{E}\Sigma_n\boldsymbol{E}^T)=  \mathrm{tr}(\boldsymbol{E}^T\boldsymbol{E}\Sigma_n).
%\end{eqnarray}
 Then 
\begin{eqnarray}
\mathrm{tr}(\Phi^2)=m+n+4\mathrm{tr}(E^TE\Sigma_n)]+\mathrm{tr}((\Sigma_n^{1/2}E^TE\Sigma_n^{1/2})^2).
\end{eqnarray} 
We saw in (\ref{bill_111_D_111}) that $\mathrm{tr}(E^TE\Sigma_n)\leq(k+1)\|\underline{r}\|_2^2\mathrm{tr}(\Sigma_n)\leq (k+1)nP\|\underline{r}\|_2^2$. The term $\mathrm{tr}((\Sigma_n^{1/2}E^TE\Sigma_n^{1/2})^2)$ can be bounded as
\begin{eqnarray}
\label{vc0}
\mathrm{tr}((\Sigma_n^{1/2}E^TE\Sigma_n^{1/2})^2)&\stackrel{(a)}{=}&\|\Sigma_n^{1/2}E^TE\Sigma_n^{1/2}\|_2^2\notag\\
&\stackrel{(b)}{\leq}&\|\Sigma_n^{1/2}E^TE\|^2\|\Sigma_n^{1/2}\|_2^2\notag\\
&\stackrel{(c)}{\leq}&\|E^TE\|^2\|\Sigma_n^{1/2}\|^2\|\Sigma_n^{1/2}\|_2^2\notag\\
&\stackrel{(d)}{=}&\|E^TE\|^2\lambda_{\max}(\Sigma_n)\mathrm{tr}(\Sigma_n)\notag\\
&\stackrel{(e)}{\leq}&nP \|E^TE\|^2\lambda_{\max}(\Sigma_n),
\end{eqnarray} 
where $(a)$ follows the definition of the norm $\|\cdot\|_2$, $(b)$ is due to Lemma~\ref{lem_1}, $(c)$ uses the fact that the matrix norm $\|\cdot\|$ is sub-multiplicative, $(d)$ is due to $\|\Sigma_n^{1/2}\|^2=\lambda_{\max}(\Sigma_n)$ and $\|\Sigma_n^{1/2}\|_2^2=\mathrm{tr}(\Sigma_n)$ and $(e)$ is due to $\mathrm{tr}(\Sigma_n)\leq nP$.  
We need the next lemma in order to find an upper bound on $\|E^TE\|^2$:
\begin{lem}
\label{lem_6}
For a matrix $M$, define the maximum row-sum norm by  
  \begin{eqnarray}
  \label{bopl_111}
\|M\|_{r}:=\max_{i}\sum_{j}|M_{i,j}|.
\end{eqnarray}
 Then $\|\cdot\|_{r}$ is a matrix norm. In particular, it is sub-multiplicative, i.e., 
 \begin{eqnarray}
 \label{sub_mul}
\|M_1M_2\|_{r}\leq \|M_1\|_{r}\|M_2\|_{r},
\end{eqnarray}
 for every two matrices $M_1$ and $M_2$ of proper sizes. Moreover, 
  \begin{eqnarray}
  \label{sonic}
\|M\|^2 \leq\|M^TM\|_{r}.
%|\lambda_{\max}(M)|\leq \|M\|_{r},
\end{eqnarray}
for every matrix $M$. 
\end{lem}
\begin{proof}
See \cite{Horn}. The statement in (\ref{sonic}) follows from Theorem 5.6.9 in \cite{Horn} which states that the size of every eigenvalue of a square matrix is bounded from above by every matrix norm of that matrix.  %To briefly describe the first inequality in (\ref{shed_11}), note that $\|M\|=(\lambda_{\max}(M^TM))^{\frac{1}{2}}$. Every matrix norm serves as an upper bound on the spectral radius of a square matrix. In particular, $\lambda_{\max}(M^TM)\leq \|M^TM\|_{r}$.
\end{proof}

We get 
\begin{eqnarray}
\label{vc1}
\|E^TE\|^2\stackrel{(a)}{\leq} \|(E^TE)^2\|_{r}\stackrel{(b)}{\leq} \|E\|^2_{r}\|E^T\|^2_{r},
\end{eqnarray}
where $(a)$ uses (\ref{sonic}) and $(b)$ is due to (\ref{sub_mul}).  The matrix $ E$ is banded, i.e., $ E_{i,j}=0$
 for $i-j<0$ or $i-j>k$ and $| E_{i,j}|\leq  r_{i-j}$
 for $0\le i-j\leq	 k$. It follows that  
 \begin{eqnarray}
 \label{ino_100}
\| E\|_{r}\leq \sum_{l=0}^kr_l=r_s,\,\,\,\| E^T\|_{r}\leq \sum_{l=0}^kr_l=r_s.
\end{eqnarray}
By (\ref{vc1}) and (\ref{ino_100}), 
\begin{eqnarray}
\label{vc2}
\|E^TE\|^2\leq r_s^4.
\end{eqnarray}
  Then (\ref{vc0}) and (\ref{vc2}) give 
\begin{eqnarray}
\mathrm{tr}((\Sigma_n^{1/2}E^TE\Sigma_n^{1/2})^2)\leq nPr_s^4\lambda_{\max}(\Sigma_n).
\end{eqnarray} 
We have shown that 
\begin{eqnarray}
\label{h00}
\mathrm{tr}(\Phi^2)\leq m+n+4(k+1)nP\|\underline{r}\|_2^2+nPr_s^4\lambda_{\max}(\Sigma_n).
\end{eqnarray}
%Since $\lambda_{\max}(\Sigma_n)\leq \sqrt{\mathrm{tr}(\Sigma_n^2)}=\|\Sigma_n\|_2$, we get $\mathrm{tr}(\Phi^2)\leq C_n/2$ as promised where $C_n$ is defined in~(\ref{cons_C}).

\section*{Appendix~D; Proof of Lemma~\ref{lem_1}}
We verify (\ref{yadam_1}) first. We have 
\begin{eqnarray}
\label{tang_1}
\|M_1M_2\|^2_2=\mathrm{tr}(M_1M_2(M_1M_2)^T). 
\end{eqnarray}
Denote the columns of $M_2$ by $\underline{v}_1,\cdots, \underline{v}_m$. Then the columns of $M_1M_2$ are $M_1\underline{v}_1,\cdots, M_2\underline{v}_m$ and we get 
\begin{eqnarray}
\label{tang_2}
M_1M_2(M_1M_2)^T=\sum_{i=1}^m M_1\underline{v}_i(M_1\underline{v}_i)^T.
\end{eqnarray}
By (\ref{tang_1}) and (\ref{tang_2}), 
\begin{eqnarray}
\|M_1M_2\|^2_2&=&\mathrm{tr}\Big(\sum_{i=1}^m M_1\underline{v}_i(M_1\underline{v}_i)^T\Big)\notag\\
&=&\sum_{i=1}^m \mathrm{tr}\big(M_1\underline{v}_i(M_1\underline{v}_i)^T\big)\notag\\
&=&\sum_{i=1}^m \|M_1\underline{v}_i\|_2^2\notag\\
&\stackrel{(a)}{\leq} &\sum_{i=1}^m \|M_1\|^2\|\underline{v}_i\|_2^2\notag\\
&=&\|M_1\|^2 \sum_{i=1}^m \|\underline{v}_i\|_2^2\notag\\
&\stackrel{(b)}{=}&\|M_1\|^2\|M_2\|_2^2,
\end{eqnarray}
where $(a)$ is due to the definition of the operator norm and $(b)$ is due to $\|M_2\|_2^2=\sum_{i=1}^m\|\underline{v}_i\|_2^2$. This proves (\ref{yadam_1}). Using the facts that $\|M\|=\|M^T\|$ and $\|M\|_2=\|M^T\|_2$ for every matrix $M$, we get 
\begin{eqnarray}
\|M_1M_2\|_2= \|(M_1M_2)^T\|_2=\|M_2^TM_1^T\|_2\stackrel{(*)}{\leq}\|M_2^T\|\|M_1^T\|_2=\|M_2\|\|M_1\|_2,
\end{eqnarray}
where $(*)$ is due to (\ref{yadam_1}) we just verified. This completes the proof of (\ref{yadam_2}).

\section*{Appendix~E; Proof of $\|H_c\|\leq \beta$ in (\ref{mahtab})}
We need the following lemma which is a restatement of Lemma~4.1 in \cite{Gray}:
\begin{lem}
\label{lem_6_5}
Let $T$ be an $n\times n$ symmetric Toeplitz matrix with real entries given by $T_{i,j}=t_{i-j}$ where $t_{-l}=t_l$ for every $0\leq l\leq n-1$. Moreover, define the function $f_T(\cdot)$ by 
\begin{eqnarray}
f_T(\omega)=\sum_{l=-(n-1)}^{n-1}t_le^{\sqrt{-1}\,l\omega},\,\,\,0\leq \omega\leq 2\pi.
\end{eqnarray}
Then every eigenvalue of $T$ lies in the interval $[\min_{\omega}f_T(\omega), \max_{\omega}f_T(\omega)]$. 
\end{lem} 
The matrix $H_c^TH_c$ is a symmetric Toeplitz matrix with entries $[H_c^TH_c]_{i,j}=\sum_{l=0}^{k-|i-j|}c_lc_{l+|i-j|}\mathds{1}_{|i-j|\leq k}$. Then $f_{H_c^TH_c}(\omega)=\big|\sum_{i=0}^{k} c_ie^{\sqrt{-1}\,i\omega}\big|^2$. By Lemma~\ref{lem_6_5}, every eigenvalue $\lambda$ of $H_c^TH_c$ satisfies 
\begin{eqnarray}
\label{laa_22}
\alpha^2\leq \lambda \leq \beta^2,
\end{eqnarray}
where $\alpha$ and $\beta$ are given in (\ref{min_max}). Since $\|H_c\|^2=\l\lambda_{\max}(H^T_cH_c)$, then (\ref{laa_22}) gives $\|H_c\|\leq \beta$.
%Moreover, 
%\begin{eqnarray}
%\|\Omega^{-1}_c\|=\lambda_{\max}(\Omega^{-1}_c)=\frac{1}{\lambda_{\min}(\Omega_c)}=\frac{1}{1+P\lambda_{\min}(H_cH_c^T)}=1+P\lambda_{\max}(H_c^TH_c)=1+P\|H_c\|^2.
%\end{eqnarray}
%Hence, 
%\begin{eqnarray}
%1+\alpha^2P\leq \|\Omega_c\|\leq 1+\beta^2P
%\end{eqnarray}E
\section*{Appendix~F; Proof of $\|E\|\leq r_s$ in (\ref{mahtab})}
The proof of $\|E\|\leq r_s$ follows similar lines of reasoning that led to $\|E^TE\|\leq r_s^2$ in (\ref{vc2}). We present the details here for completeness. By Lemma~\ref{lem_6}, 
\begin{eqnarray}
\label{goosheh}
\|E\|\leq (\|E\|_{r}\|E^T\|_{r})^{1/2},
\end{eqnarray}
where $\|\cdot\|_{r}$ is the maximum row-sum norm defined in (\ref{bopl_111}). The matrix $ E$ is banded, i.e., $ E_{i,j}=0$
 for $i-j<0$ or $i-j>k$ and $| E_{i,j}|\leq r_{i-j}$
 for $0\le i-j\leq	 k$. It follows that both $\|E\|_{r}$ and $\|E^T\|_{r}$ are not larger than $\sum_{l=0}^k r_l=r_s$. Then (\ref{goosheh}) implies that $\|E\|$ is also not larger than $r_s$. 

\section*{Appendix~G; Proof of (\ref{c_p_n})}
By (\ref{psi_09}),  
 \begin{eqnarray}
 \label{sard_0}
\mathrm{tr}(\Psi^2)=\mathrm{tr}((\Sigma_n^{1/2}H^T\Omega_c^{-1}H\Sigma_n^{1/2})^2)+2\mathrm{tr}(\Sigma_n^{1/2}H^T\Omega_c^{-2}H\Sigma_n^{1/2})+\mathrm{tr}(\Omega_c^{-2}).
\end{eqnarray}
Since each eigenvalue of $\Omega_c^{-1}$ is at most $1$, we have 
\begin{eqnarray}
\label{sard_1}
\mathrm{tr}(\Omega_c^{-2})\leq m.
\end{eqnarray}
Moreover, 
\begin{eqnarray}
\label{pahloo_1}
\mathrm{tr}((\Sigma_n^{1/2}H^T\Omega_c^{-1}H\Sigma_n^{1/2})^2)&=&\|\Sigma_n^{1/2}H^T\Omega_c^{-1}H\Sigma_n^{1/2}\|_2^2\notag\\
&\stackrel{(a)}{\leq}&\|\Sigma_n^{1/2}H^T\Omega_c^{-1}H\|^2\|\Sigma_n^{1/2}\|_2^2\notag\\
&\stackrel{(b)}{\leq}&\|H^T\|^2\,\|\Omega_c^{-1}\|^2\,\|H\|^2\|\, \|\Sigma_n^{1/2}\|^2\,\|\Sigma_n^{1/2}\|_2^2\notag\\
&\stackrel{(c)}{\leq}&\|H\|^4\|\, \lambda_{\max}(\Sigma_n)\mathrm{tr}(\Sigma_n)\notag\\
&\stackrel{(d)}{\leq}&nP(\beta+r_s)^4\lambda_{\max}(\Sigma_n),
%&\stackrel{(e)}{\leq}&nP(\beta+r_s)^4\|\Sigma_n\|_2,
\end{eqnarray}
where $(a)$ is due to Lemma~\ref{lem_1}, $(b)$ is due to $\|\cdot\|$ being sub-multiplicative, $(c)$ is due to $\|H\|=\|H^T\|$, $\lambda_{\max}(\Sigma_n)=\|\Sigma_n^{1/2}\|^2$ and $\mathrm{tr}(\Sigma_n)=\|\Sigma_n^{1/2}\|_2^2$ and  $(d)$ is due to $\mathrm{tr}(\Sigma_n)\leq nP$ and $\|H\|\leq \|H_c\|+\|E\|\leq \beta+r_s$. This last inequality itself is a consequence of the second and third inequalities in (\ref{mahtab}). Following similar lines of reasoning as in (\ref{pahloo_1}), 
\begin{eqnarray}
\label{pahloo_2}
\mathrm{tr}(\Sigma_n^{1/2}H^T\Omega_c^{-2}H\Sigma_n^{1/2})&=&\|\Sigma_n^{1/2}H^T\Omega_c^{-1}\|_2^2\notag\\
&\stackrel{}{\leq}& \|H^T\Omega_c^{-1}\|^2\|\Sigma_n^{1/2}\|_2^2\notag\\
&\leq& \|H^T\|^2\|\Omega_c^{-1}\|^2\,\|\Sigma_n^{1/2}\|_2^2\notag\\
&\leq&\|H\|^2\mathrm{tr}(\Sigma_n)\notag\\
&\leq&(\beta+r_s)^2nP.
\end{eqnarray}
By (\ref{sard_0}), (\ref{sard_1}), (\ref{pahloo_1}) and (\ref{pahloo_2}), 
\begin{eqnarray}
\mathrm{tr}(\Psi^2)\leq m+2(\beta+r_s)^2nP+nP(\beta+r_s)^4\lambda_{\max}(\Sigma_n).
\end{eqnarray}

\section*{Appendix~H; Proof of Lemma~\ref{seal} on the size of the Gaussian typical set }
For $r>0$, positive integer $n$ and real positive-definite matrix $\Sigma$,  consider the ellipsoid 
\begin{eqnarray}
\label{mehr_456}
\mathcal{E}_r:=\Big\{\underline{a}\in \mathbb{R}^{n}: \frac{1}{n}\underline{a}^T\Sigma^{-1}\underline{a}<r.\Big\}.
\end{eqnarray} 
Let $\Sigma=U\Lambda U^T$ be the spectral decomposition for $\Sigma$ where $U$ is a real orthogonal matrix and $\Lambda$ is a diagonal matrix whose diagonal entries $\lambda_1,\cdots, \lambda_n$ are the eigenvalues of $\Sigma$.  The set $U^T\mathcal{E}_r:=\{U^T\underline{a}: \underline{a}\in \mathcal{E}_r\}$ is an ellipsoid in standard from, i.e.,  
\begin{eqnarray}
U^T\mathcal{E}_r=\Big\{\underline{a}\in \mathbb{R}^n: \frac{1}{n}\sum_{i=1}^{n}\frac{a^2_i}{\lambda_i}<r\Big\}.
\end{eqnarray}
The volume of the standard ellipsoid $\big\{(a_1,\cdots, a_n): \sum_{i=1}^n \frac{a_i^2}{c_i^2}\le1\big\}$ is $\frac{\pi^{\frac{n}{2}}}{\Gamma(\frac{n}{2}+1)}\prod_{i=1}^n c_i$ where $\Gamma(\cdot)$ is the Gamma function. Hence,
\begin{eqnarray}
\label{vol_elip}
\mathrm{Vol}(\mathcal{E}_r)&\stackrel{(a)}{=}&\mathrm{Vol}(U^T\mathcal{E}_r)\notag\\
&=&\frac{\pi^{\frac{n}{2}}}{\Gamma(\frac{n}{2}+1)}\prod_{i=1}^{n}(nr\lambda_i)^{\frac{1}{2}}\notag\\&\stackrel{(b)}{=}&\frac{\pi^{\frac{n}{2}}n^{\frac{n}{2}}}{\Gamma(\frac{n}{2}+1)}\big(\det(\Sigma)\big)^{\frac{1}{2}}r^{\frac{n}{2}}\notag\\
&\stackrel{(c)}{\leq}&(2\pi e)^{\frac{n}{2}}\big(\det(\Sigma)\big)^{\frac{1}{2}}r^{\frac{n}{2}}\notag\\
&=&2^{\frac{1}{2}\log((2\pi e)^n\det(\Sigma))+\frac{n}{2}\log r}\notag\\
&=&2^{h_{\mathsf{G}}(\Sigma)+\frac{n}{2}\log r},
\end{eqnarray}
where $(a)$ is due to the fact that the map $\underline{a}\mapsto U^T\underline{a}$ is an isometry, $(b)$ is due to $\det (\Sigma_n)=\prod_{i=1}^n\lambda_i$ and $(c)$ is due to the inequality $\Gamma(x+1)\geq x^x e^{-x}$ for every $x>0$ according to Corollary~1.2 in \cite{Gamma}. Finally, 
\begin{eqnarray}
\label{ema_1}
\mathrm{Vol}(\mathcal{T}^{(n)}_\eta(\Sigma))&=&\mathrm{Vol}(\mathcal{E}_{1+\eta})-\mathrm{Vol}(\mathcal{E}_{1-\eta})1_{\eta<1}\notag\\&\leq&\mathrm{Vol}(\mathcal{E}_{1+\eta})\notag\\&\leq&2^{h_{\mathsf{G}}(\Sigma)+\frac{n}{2}\log(1+\eta)},
\end{eqnarray}
where $1_{\eta<1}$ is $1$ if $\eta<1$ and is $0$ otherwise. Corollary~1.2 in \cite{Gamma} also implies that $\Gamma(x+1)\leq x^x e^{-x} \sqrt{2\pi(x+1)}$ for every $x>0$. Using this inequality in (\ref{vol_elip}), we get 
\begin{eqnarray}
\label{nasl_1}
\mathrm{Vol}(\mathcal{E}_r)\geq 2^{h_{\mathsf{G}}(\Sigma)+\frac{n}{2}\log r-\frac{1}{2}\log(\pi(n+2))}.
\end{eqnarray}
If $\eta\geq 1$, then (\ref{nasl_1}) gives
\begin{eqnarray}
\label{ema_2}
\mathrm{Vol}(\mathcal{T}^{(n)}_\eta(\Sigma))&=&\mathrm{Vol}(\mathcal{E}_{1+\eta})\notag\\
&\geq& 2^{h_{\mathsf{G}}(\Sigma)+\frac{n}{2}\log(1+\eta)-\frac{1}{2}\log(\pi(n+2))}.
%&=&2^{h_{\mathsf{G}}(\Sigma)+\frac{n}{2}\log(1+\eta)}\big(2^{-\frac{1}{2}\log(\pi(n+2))}-2^{\frac{n}{2}\log\frac{1-\eta}{1+\eta}}1_{\eta<1}\big).
\end{eqnarray}
Finally, (\ref{ema_1}) and (\ref{ema_2}) verify (\ref{typ_asy}) for every $\eta\geq 1$.

  \section*{Appendix~I; Proof of Lemma~\ref{lem_55}}
 Fix $H\in \mathcal{H}$. For a real symmetric $n\times n$ matrix $X$,  denote its eigenvalues in increasing order by $\lambda_i(X)$ for $1\le i\le n$, i.e., $\lambda_1(X)\leq \lambda_2(X)\le\cdots\le \lambda_n(X)$.  We need the following lemma which is a direct consequence of Theorem~4.3.1 in \cite{Horn} due to Hermann~Weyl: 
 \begin{lem}
 \label{lem_9}
 Let $A$ and $B$ be $n\times n$ real symmetric matrices. Then 
 \begin{eqnarray}
|\lambda_{i}(A)-\lambda_i(B)|\leq \|A-B\|,
\end{eqnarray}
for every $1\leq i\leq n$.
 \end{lem}
 Denote the eigenvalues of $H^TH\Sigma_n $ and $H_c^TH_c\Sigma_n$ in increasing order by $\mu_i$~and~$\lambda_i$,~respectively. We have\footnote{ The second step uses the identity $\det(I_m+AB)=\det(I_n+BA)$ for matrices $A$ and $B$ of sizes $m\times n$ and $n\times m$, respectively.}
 \begin{eqnarray}
 \label{yar_11}
\frac{\det(\Omega_H)}{\det(\Omega_c)}=\frac{\det(I_m+H\Sigma_n H^T)}{\det(I_m+H_c\Sigma_n H_c^T)}=\frac{\det(I_n+H^TH\Sigma_n)}{\det(I_n+H_c^TH_c\Sigma_n)}=\prod_{i=1}^n\frac{1+\mu_i}{1+\lambda_i}.
\end{eqnarray}
 Let us write
\begin{eqnarray}
\label{yar_22}
\frac{1+\mu_i}{1+\lambda_i}=1+\frac{\mu_i-\lambda_i}{1+\lambda_i}\geq 1-\frac{|\mu_i-\lambda_i|}{1+\lambda_{\min}(H_c^TH_c\Sigma_n)},
\end{eqnarray} 
for every $1\le i\leq n$ where $\lambda_{\min}(H_c^TH_c\Sigma_n)=\min_{1\leq i\leq n}\lambda_i$. %By (\ref{laa_22}), $\min_i\lambda_i\geq \alpha^2$ where $\alpha$ is given in~(\ref{min_max}). 
By Lemma~\ref{lem_9}, 
\begin{eqnarray}
\label{yar_333}
|\mu_i-\lambda_i|\leq \|H^TH\Sigma_n -H_c^TH_c\Sigma_n\|\leq \|H^TH-H_c^TH_c\|\|\Sigma_n\|.
\end{eqnarray}
If $\alpha>0$, then by (\ref{laa_22}) in Appendix~E, every eigenvalue of $H_c^TH_c$ is positive and hence, this matrix is invertible. In this case,  
\begin{eqnarray}
\label{yar_444}
\lambda_{\min}(H_c^TH_c\Sigma_n)&=&\frac{1}{\lambda_{\max}((H_c^TH_c\Sigma_n)^{-1})}\notag\\
&\stackrel{}{=}&\frac{1}{\lambda_{\max}(\Sigma_n^{-1}(H_c^TH_c)^{-1})}\notag\\
&\stackrel{(a)}{\geq}&\frac{1}{\lambda_{\max}(\Sigma_n^{-1})\lambda_{\max}((H_c^TH_c)^{-1})}\notag\\
&=&\lambda_{\min}(\Sigma_n)\lambda_{\min}(H_c^TH_c)\notag\\
&\stackrel{(b)}{\geq}&\alpha^2\lambda_{\min}(\Sigma_n),
\end{eqnarray}
where $(a)$ uses the identity\footnote{This identity is part of Problem III.6.14 on page 78 in \cite{Bhatia}.} $\lambda_{\max}(AB)\leq \lambda_{\max}(A)\lambda_{\max}(B)$ for positive semidefinite matrices $A$ and $B$ and $(b)$ is a direct consequence of (\ref{laa_22}) in Appendix~E.  By (\ref{yar_22}), (\ref{yar_333}) and (\ref{yar_444}),  
\begin{eqnarray}
\label{yar_33}
\frac{1+\mu_i}{1+\lambda_i}\geq 1-\frac{\|H^TH-H_c^TH_c\|\|\Sigma_n\|}{1+\alpha^2\lambda_{\min}(\Sigma_n)}.
\end{eqnarray} 
Recall the matrix $E$ in~(\ref{err_111}). Then
\begin{eqnarray}
\label{yar_44}
\|H^TH-H_c^TH_c\|=\|E^TE+E^TH_c+H_c^TE\|\leq \|E\|^2+2\|E\|\|H_c\|\leq  r_s^2+2 r_s\beta,
\end{eqnarray}
where the last step is due to $\|H_c\|\leq \beta$ and $\|E\|\leq r_s$ proved in Appendix~E and Appendix~F, respectively.  By (\ref{yar_11}), (\ref{yar_33}) and (\ref{yar_44}), 
\begin{eqnarray}
\det(\Omega_H)\ge \det(\Omega_c)\Big(1-\frac{( r^2_{s}+2r_s\beta)\|\Sigma_n\|}{1+\alpha^2\lambda_{\min}(\Sigma_n)}\Big)^n.
\end{eqnarray}
This concludes the proof of (\ref{zcb_11}).

\section*{Appendix~J;  Proof of Lemma~\ref{lem_66}}
We write 
\begin{eqnarray}
\label{dast_11}
\inf_{\underline{y}\in \mathcal{T}^{(m)}_{\eta'}(\Omega_c)}\underline{y}^T\Omega_H^{-1}\underline{y}=\inf_{m(1-\eta')< t< m(1+\eta')}\,\,\,\,\,\min_{\underline{y}\in\mathbb{R}^m:\underline{y}^T\Omega_c^{-1}\underline{y}=t}\underline{y}^T\Omega_H^{-1}\underline{y}.
\end{eqnarray}
To achieve $\min_{\substack{\underline{y}\in\mathbb{R}^m\\\underline{y}^T\Omega_c^{-1}\underline{y}=t}}\underline{y}^T\Omega_H^{-1}\underline{y}$, the gradient of $\underline{y}^T\Omega_H^{-1}\underline{y}$ must be parallel to the gradient of $\underline{y}^T\Omega_c^{-1}\underline{y}$, i.e., 
\begin{eqnarray}
\label{how_11}
\Omega_H^{-1}\underline{y}=\lambda \Omega_c^{-1}\underline{y},
\end{eqnarray}
where $\lambda$ is the Lagrange multiplier. This tells us that the ``optimum'' $\underline{y}$ must be an eigenvector for the matrix $\Omega_c\Omega_H^{-1}$ with corresponding eigenvalue $\lambda$. Multiplying both sides of (\ref{how_11}) by $\underline{y}^T$ from left, we conclude that 
\begin{eqnarray}
\label{dast_22}
\inf_{\underline{y}\in\mathbb{R}^m:\underline{y}^T\Omega_c^{-1}\underline{y}=t}\underline{y}^T\Omega_H^{-1}\underline{y}=t\lambda_{\min}(\Omega_c\Omega_H^{-1}),
\end{eqnarray} 
where $\lambda_{\min}(\Omega_c\Omega_H^{-1})$ is the minimum eigenvalue of $\Omega_c\Omega_H^{-1}$. Putting (\ref{dast_11}) and (\ref{dast_22}) together,  
\begin{eqnarray}
\label{blah_1}
\inf_{\underline{y}\in \mathcal{T}^{(m)}_{\eta'}(\Omega_c)}\underline{y}^T\Omega_H^{-1}\underline{y}=m(1-\eta') \lambda_{\min}(\Omega_c\Omega_H^{-1})=\frac{m(1-\eta')}{\lambda_{\max}(\Omega_H\Omega^{-1}_c)}.
\end{eqnarray}
Next, we find an upper bound on $\lambda_{\max}(\Omega_H\Omega_c^{-1})$. Writing $H=H_c+E$, we get 
\begin{eqnarray}
\label{vrty_1}
\Omega_H\Omega_c^{-1}&=&(\underbrace{I_m+H_c\Sigma_n H_c^T}_{=\Omega_c}+H_c\Sigma_n E^T+E\Sigma_n H_c^T+E\Sigma_nE^T)\Omega_c^{-1}\notag\\
&=&I_m+H_c\Sigma_n E^T\Omega_c^{-1}+E\Sigma_n H_c^T\Omega_c^{-1}+E\Sigma_nE^T\Omega_c^{-1}.
\end{eqnarray}
Then 
\begin{eqnarray}
\label{blah_123}
\lambda_{\max}(\Omega_H\Omega_c^{-1})&\stackrel{(a)}{\leq}& \|\Omega_H\Omega_c^{-1}\|\notag\\
&\stackrel{(b)}{\leq}&1+2\|H_c\|\|E\||\Sigma_n\|+\|E\|^2\|\Sigma_n\|\notag\\
&\stackrel{(c)}{\leq}&1+r_s(r_s+2\beta)\|\Sigma_n\|\notag\\
&\stackrel{(d)}{=}&1+r_s(r_s+2\beta)\lambda_{\max}(\Sigma_n),
\end{eqnarray} 
where $(a)$ follows from Theorem 5.6.9 in \cite{Horn} which states that the size of every eigenvalue of a square matrix is bounded from above by every matrix norm of that matrix, $(b)$ is due to (\ref{vrty_1}) and the properties of the operator norm and $(c)$ is due to $\|H_c\|\le \beta$ and $\|E\|\leq r_s$ proved in Appendix~E and Appendix~F, respectively, and $(d)$ is due to $\lambda_{\max}(M)=\|M\|$ for every symmetric matrix $M$. Finally, (\ref{blah_1}) and (\ref{blah_123}) give the promised bound in (\ref{zcb_22}).

\end{document}